\documentclass[12pt]{article}
\usepackage{amsmath,amsxtra,amssymb,latexsym, amscd,amsthm,pb-diagram,fancyhdr,euscript}
\usepackage{amsthm}
\usepackage{latexsym}
\usepackage{amssymb}
\usepackage{amsmath}
\usepackage{indentfirst}
\usepackage{hyperref}
\usepackage{mathrsfs}
\usepackage{array}
\usepackage{euscript}
\usepackage{slashed}
\usepackage[T1]{fontenc}
\usepackage{graphicx}
\usepackage[utf8]{inputenc}
\usepackage[french,english]{babel}
\usepackage{lipsum} 
\usepackage{multicol}
\hypersetup{
   colorlinks,
    citecolor=red,
    filecolor=black,
    linkcolor=blue,
    urlcolor=magenta
}
\hypersetup{linktocpage}

\newtheorem{theorem}{Theorem}[section]
\newtheorem{lemma}[theorem]{Lemma}
\newtheorem{proposition}[theorem]{Proposition}

\newtheorem{definition}[theorem]{Definition}

\newtheorem{remark}[theorem]{Remark}

\newcommand{\norm}[1]{\left\Vert#1\right\Vert}

\numberwithin{equation}{section}
\usepackage{amsfonts}

\usepackage{amsmath}
\usepackage{amssymb}
\usepackage{cases}

\def\dive{\mathrm{div}}
\def\cal{\mathcal}

\def\curl{\mathrm{Curl}}

\def\geq{\geqslant} 
\def\leq{\leqslant}

\def\C+{C_+([t_0,\infty))}

\begin{document}\mbox{}
\vspace{0.25in}

\begin{center}

{\huge{\bf The simplified Bardina equation on two-dimensional closed manifolds}}

\vspace{0.25in}

\large{{\bf PHAM Truong Xuan}\footnote{Faculty of Information Technology, Department of Mathematics, Thuyloi university, Khoa Cong nghe Thong tin, Bo mon Toan, Dai hoc Thuy loi, 175 Tay Son, Dong Da, Ha Noi, Viet Nam. 

\noindent
Email: xuanpt@tlu.edu.vn or phamtruongxuan.k5@gmail.com}}

\end{center}

\begin{abstract} 
In this paper we study the viscous simplified Bardina equation on the two-dimensional closed manifold $M$ which is embedded in $\mathbb{R}^3$. First, we prove the existence and the uniqueness of the weak solutions and also the existence of the global attractor for the equation on $M$. Then we establish the upper and lower bounds of the Hausdorff and fractal dimensions of the attractor. We also prove the existence of an inertial manifold for the equation on the two-dimensional sphere ${S}^2$.
\end{abstract}




\tableofcontents

\section{Introduction}
Since the existing mathematical theory is not sufficient to prove the global well-posedness of the $3$D Navier-Stokes equations (NSE), the dynamics of homogeneous incompressible fluid flows are not known so far. The mathematicians study these dynamics by using the direct numerical simulation of NSE and consider the mean characteristics of the flow by averaging techniques in many practical applications (see for example \cite{Mars2000,Mars2001,Mars2003}). This leads to the well-known closure problem i.e the following Reynolds averaged NSE is not closed (see \cite{CaLuTi}).   
\begin{eqnarray}\label{RANS}
\bar{v}_t - \nu \Delta \bar{v} + \nabla \cdot (\overline{v\otimes v}) &=& - \nabla \bar{p} + \bar{f}, \nonumber\\
\nabla \cdot \bar{v} &=& 0. 
\end{eqnarray}
Here we can write
\begin{equation*}
\nabla \cdot (\overline{v\otimes v}) = \nabla \cdot (\bar{v}\otimes \bar{v}) + \nabla \cdot \mathcal{R}(v,v),\\
\end{equation*}
with $\mathcal{R}(v,v) = \overline{v\otimes v} - \bar{v}\otimes \bar{v}$ is the Reynolds stress tensor.
However, on the turbulence modeling applications, one need to produce simplified, reliable and computationally realizable closure models. For this reason, in order to obtain the closure models Bardina et al. \cite{BaFe} modified the Reynolds stress tensor by
\begin{equation*}
\mathcal{R}(v,v) \simeq \overline{\bar{v}\otimes \bar{v}} - \bar{\bar{v}} \otimes \bar{\bar{v}}.
\end{equation*}
After that, Layton and Lewandowski \cite{LaLe} considered a simpler form of the above approximation of the Reynolds stress tensor
\begin{equation*}
\mathcal{R}(v,v) \simeq \overline{\bar{v}\otimes \bar{v}} - \bar{v} \otimes \bar{v}.
\end{equation*}
The modification of Layton and Lewandowski leads to study the following sub-grid scale turbulence model
(or called simplified Bardina equation)
\begin{eqnarray}\label{MRANS}
{\omega}_t - \nu \Delta \omega + \nabla \cdot (\overline{\omega\otimes \omega}) &=& - \nabla q + \bar{f}, \cr
\nabla \cdot \omega &=& 0, \cr
\omega(x,0) &=& \bar{v}_0(x), 
\end{eqnarray}
where $(\omega,q)$ is an approximation of $(\bar{v},\bar{p})$. Following \cite{LaLe}, the simplified Bardina equation is considered with the filtering kernel associated with the Helmholtz operator $(I-\alpha^2\Delta)^{-1}$. This means that if $v$ is the unfiltered velocity and $u=\omega$ is the smooth filtered velocity then $v=u-\alpha^2\Delta u$ and also keep that $p=q-\alpha^2\Delta q$, then the equation\eqref{MRANS} becomes
\begin{eqnarray}\label{SBRANS}
v_t - \nu \Delta v + (u \cdot \nabla)u &=& - \nabla p + f, \cr
\nabla u &=& \nabla v= 0, \cr
v &=& u - \alpha^2\Delta u \cr
u(x,0) &=& u^{in}(x), 
\end{eqnarray}
where $u$ and $v$ are periodic with periodic box $\Omega = [0,2\pi L]^3$.

The global existence and uniqueness of weak solutions of the equation \eqref{SBRANS} with the periodic boundary conditions in three-dimension is established early by Layton and Lewandowski \cite{LaLe} and then expanded to study by Titi et al \cite{CaLuTi}. In detail, the last work has proven the global well-posedness for weaker initial conditions than the first work, then considered the upper bound to the dimension of the global attractor and given the relation between the modified Bardina equation and the modified Euler equation. The existence of inertial manifold for the simplified Bardina equation is studied by Titi et al. in \cite{Titi2014} in the two-dimension with periodic boundary condition case. On the other hand, there are many works about the other turbulence models such as the modified-Leray-$\alpha$ and viscous Camassa-Holm or Navier-Stokes-$\alpha$ on the same framework, see for example \cite{Che,CheHoOlTi,Fo,Ily2004,Il2004',IlLuTi}.

The Navier-Stokes equation and the turbulence equations are studied on the generalized compact Riemannian manifolds in the works of Ebin and Marsden \cite{EbiMa}, Skholler \cite{Shko1998,Shko2000} and Skholler et al. \cite{Mars2000} via the geometry and the analysis of group of diffeomorphisms. In the specific compact manifolds such as two-dimensional sphere and square torus, the Navier-Stokes equation was studied in the works of Ilyin \cite{Il1990,Il1992,Il1988,Il1994,Il1999} and developed recently by Ilyin, Laptev and Zelik \cite{Ily2018,Ily2019,Ily2020}. In these works, they proved the well-posedness of the weak solution, then estimated the upper bound of the Hausdorff and fractal dimensions of the global attractor. For the turbulence equations Ilyin and Titi studied the attractor of the modified-Leray-$\alpha$ equation on the two-dimensional sphere and the square torus \cite{Il2004'}. In detail, they established the upper and lower bounds dependeding on $\alpha$ of the Hausdorff and fractal dimensions of the attractor. The method is based on the vorticity scalar form (see also \cite{Il1999} for the Navier-Stokes equation) of the model and the theorem about the relation between the Lyapunov exponents and the Hausdorff (fractal) dimension of attractor (see \cite{Il2001,Il2004,Te1988}). Another important technique is used to estimate the attractor's dimensions that is the Lieb-Sobolev-Thirring inequality (see \cite{Te1984}). It plays an important role to estimate the Lyapunov exponents. The Sobolev-Lieb-Thirring inequality on manifolds is considered initially by Teman et al. in \cite{Te88}, then it is improved by Ilyin et al. in the recent works on the sphere and torus \cite{Ily2019,Ily2020}. Furthermore, by considering the Navier-Stokes equation on the domain of sphere, Ilyin and Laptev \cite{Ily2018} improved the Berezin-Li-Yau inequality on the lower of the sum of the eigenvalues and therefore obtain the upper of the dimension of attractor.  

In the present paper we study the simplified Bardina equation \eqref{SBRANS} on a two-dimensional closed manifold. More precisely we study the existence and uniqueness of the weak solutions, estimate of the Hausdorff and fractal dimension of attractor and the existence of the inertial manifold. Since on a two-dimensional closed manifold there is Kodaira-Hodge decomposition of the space of smooth vector fields with the appearance of the harmonic functions, we need to add some dissipative term to the original equation to obtain a dissipative system (see Section 2.2). Then the global well-posedness will be done by the Garlekin approximation scheme and note to control the norms of the harmonic functions. We develop the methods in \cite{Il1999,Il2001,Il2004} to establish the upper bound of the Hausdorff and fractal dimensions of the global attractor for Equation \eqref{SBRANS}. Then we will develop the methods in \cite{Liu,Il2004',Il2021} to find the lower bound of the attractor's dimensions on the two-dimensional torus. In particular, we  construct a family of stationary solutions arising from the family of Kolmogorov flows and establish the lower bound for the dimension of the unstable manifold around these stationary solutions. As a consequence we obtain the lower bound of the global attractor's dimensions.
The existence of the inertial manifold is proven by feature of the spectral of Laplacian operator on two-dimensional sphere $S^2$ and the estimates of the nonlinear parts via the appearance of the alpha parameter.
  
This paper is organized as follows:
Section 2 gives the setting of the simplified Bardina equations on the generalized two-dimensional closed manifolds $M$.
Section 3 we establish the global well-posedness of the simplified Bardina equation on $M$. Section 4 we establish the upper bound of the Hausdorff and fractal dimensions of the global attractor for the equation on $S^2$ then on $M$, then we establish the lower bound of the attractor's dimensions for the equation on the two-dimensional torus $T^2=[0; \, 2\pi]\times[0; \, 2\pi]$. In Section 5 we prove the existence of an inertial manifold for the equation on $S^2$.

\subsection*{Acknowledgment}
The author would like to thank Prof. Ilyin (Keldysh Institute of Applied mathematics, Russian Academy of Sicences) for many helpful discussions.

\section{Geometrical and analytical setting}
\subsection{Two-dimensional closed manifolds and functional spaces}
Let $M$ be a $2$-dimensional closed manifold embedded in $\mathbb{R}^3$. We denote by $TM$ the set of tangent vector fields on $M$ and by $(TM)^{\bot}$ the set of normal vector fields. Following \cite{Il1990,Il1988,Il1994}, we define the two operators
$$\curl_n: TM \rightarrow (TM)^{\bot} \, \mbox{and} \, \curl:(TM)^{\bot}\rightarrow TM$$
in a neighbourhood of $M$ in $\mathbb{R}^3$:
\begin{definition}
Let $u$ be a smooth vector field on $M$ with values in $TM$, and let $\vec{\psi}$ be a smooth vector field on $M$ with values in $(TM)^\bot$, i.e. $\vec{\psi} = \psi\vec{n}$, where $\vec{n}$ is the outward unit normal vector to $M$ and $\psi$ is a smooth scalar function. We then identify the vector field $\vec{\psi}$ with the scalar function $\psi$. Let $\hat u$ and $\hat\psi$ be smooth extensions of $u$ and $\psi$ into a neighbourhood of $M$ in $\mathbb{R}^3$ such that $\hat{u}|_{M}=u$ and $\hat{\psi}|_{M}=\psi$. For $x \in M$ and $y\in \mathbb{R}^3$, we define
$$\mathrm{Curl}_nu(x) = (\curl\hat{u}(y)\cdot \vec{n}(y))\vec{n}(y)|_{y=x},$$
$$\mathrm{Curl}\vec{\psi}(x) = \curl\psi(x)= \curl\hat{\psi}(y)|_{y=x},$$
where the operator $\curl$ that appears on the right hand sides is the classical $\curl$ operator in $\mathbb{R}^3$.
\end{definition} 
The above definitions of $\curl_nu$ and $\curl\psi$ are independent of the choice of the neighbourhood of $M$ in $\mathbb{R}^3$. Moreover, the following formulas hold
\begin{equation}\label{equal1}
\curl_nu = -\vec{n}\dive(\vec{n}\times u), \, \curl\psi= -\vec{n}\times \nabla \psi,
\end{equation}
\begin{equation}\label{equal2}
\nabla_uu = \nabla\frac{|u|^2}{2} - u \times \curl_nu,
\end{equation}
\begin{equation}\label{equal3}
\Delta u = \nabla \dive u - \curl\curl_n u,
\end{equation}
where $\times$ is the outer vector product in $\mathbb{R}^3$, $\nabla$ is the covariant derivative along the vector fields and $\Delta=d\delta+\delta d$ is the Hodge-Laplacian operator.

Let $L^p(M)$ and $L^p(TM)$ be the $L^p$-spaces of the scalar functions and the tangent vector fields on $M$ respectively. Let $H^p(M)$ and $H^p(TM)$ be the corresponding Sobolev spaces of scalar functions and vector fields. The inner product on $L^2(M)$ and $L^2(TM)$ are given by
$$\left<u,v\right>_{L^2(M)} = \int_{M}u\bar{v}dM, \, \mbox{for} \, u,v \in L^2(M),$$
$$\left<u,v\right>_{L^2(TM)} = \int_{M} u \cdot \bar{v} dM, \, \mbox{for} \, u,v \in L^2(TM).$$
The following integration by parts formulas will be used frequently
$$\left<\nabla h,v\right>_{L^2(TM)}=-\left<h,\dive v\right>_{L^2(M)},$$
$$\left<\curl\vec{\psi},v\right>_{L^2(TM)} = \left<\vec{\psi},\curl_nv\right>_{L^2(M)}.$$
By using Kodaira-Hodge decomposition we have
$$C^\infty(TM) = \left\{\nabla \psi \, : \, \psi \in C^\infty(M) \right\} \oplus \left\{ \mathrm{Curl}\psi \, : \, \psi \in C^\infty(M) \right\} \oplus \mathcal{H}^1,$$
where $\mathcal{H}^1$ is the finite-dimensional space of harmonic $1-$forms. Putting 
$$\mathcal{V} = \left\{ \curl\psi \, : \, \psi \in C^\infty(M) \right\} \, , \, H = \overline{\mathcal{V}}^{L^2(TM)} \, , \, V = \overline{\mathcal{V}}^{H^1(TM)},$$
endowed with the norms
$$\norm{u}^2_H = \left<u,u\right>, \, \norm{u}^2_V = \left<Au,u\right> = \left<\curl_nu, \curl_nu \right>.$$
Since $\dive u = 0$, we have the Poincar\'e inequality
\begin{equation}\label{norm1}
\norm{u}_H \leq \lambda_1^{-1/2} \left( \norm{u}_V + \norm{\dive u}_H \right) = \lambda_1^{-1/2}\norm{u}_V\end{equation}
where $\lambda_1$ is the first eigenvalue of the Stokes operator $A=\curl\curl_n$ (see the below proposition). We know that
\begin{equation}\label{norm2}
\norm{u}_{H^1(TM)} = \norm{u}^2_{L^2(TM)} + \norm{\dive u}^2_{L^2(M)} + \norm{\curl_n u}^2_{L^2(M)}.
\end{equation}
From the inequalities \eqref{norm1}, \eqref{norm2} and since $\dive u=0$ on $V$, the norms on $H^1$ and $V$ are equivalent for all  $u\in V$. In the rest of this paper, we denote $\norm{.}_{L^2}:=|.|$, $\norm{.}_{V}:=\norm{.}$ and $\norm{.}_{H^1}:=\norm{.}_1$.

Let $\mathbb{P}: L^2(TM)\rightarrow H$ be the orthogonal projection i.e Helmholtz-Leray projection on $H$, and let $A=-\mathbb{P}\Delta = -\Delta \mathbb{P} = \curl\curl_n$ be the Stokes operator with domain $D(A)=H^2(TM)\cap V$. Considering the linear Stokes problem
\begin{equation}\label{StoEq}
A u + \mathrm{grad} p =f, \, \dive u=0.
\end{equation}
Taking the inner product of this equation with $v\in V$ we get
\begin{equation*}
\left< \curl_n u, \curl_n u \right> = \left< f,u\right> \Leftrightarrow \left\| u\right\|_V = \left< f,u\right>.
\end{equation*}
By Lax-Milgram theorem, for each $f\in H^{-1}(TM)$ the weak solution of \eqref{StoEq} exists and in unique. Hence $A: H^1(TM) \longrightarrow H^{-1}(TM)$ is a linear operator with compact inverse. 
As a direct consequence, we find that problem \eqref{StoEq} has an orthonormal smooth eigenfunctions $\omega_i$
(dense in $H$ and $V$) i.e
\begin{equation*}
\curl\curl_n \omega_i = \lambda_i\omega_i, \, \dive \omega_i=0.
\end{equation*}
The relations between the eigenfunctions $\omega_i$ and the ones $\psi_i$ of the scalar Laplacian $\Delta=\dive\mathrm{grad}$ on $M$ are
\begin{equation*}
-\Delta\psi_i = \lambda_i\psi_i, \, \omega_i = n\times \mathrm{grad}\psi_i = -\curl \psi_i.
\end{equation*}
We summarize the properties of the Stokes operator $A$ in the following proposition
\begin{proposition}
The operator $A=\curl\curl_n$ is unbounded, positive, self-adjoint, symmetric in $H$ with eigenvalues $0<\lambda_1\leq\lambda_2\leq...$ which is only accumulation point $+\infty$. Moreover, its eigenvalues correspond to an orthonormal basis in $H$ (which is also orthogonal in $V$).
\end{proposition}

\subsection{The simplified Bardina equations}
In 1980 Bardina et al. \cite{BaFe} introduced a particular sub-grid scalar model which was
later simplified by Layton and Lewandowski \cite{LaLe} (therefore, we call this system by simplified Bardina equation):
\begin{align}\label{BarEqu}
\begin{cases}
v_t - \nu \Delta v + (u \cdot \nabla) u + \nabla p = f,\\
\nabla \cdot v= \nabla \cdot u=0,\\
v= u - \alpha^2 \Delta u,\\
u(0)=u_0,
\end{cases}
\end{align}
where the unknowns are the fluid velocity vector field $v$, the ''filtered'' velocity vector field $u$ and the ''filtered'' pressure scalar $p$. Besides, the constant $\nu>0$ is the kinematic viscosity coefficient and $f$ is the body force assumed to be time independent. 
On a $2$-dimension closed manifold $M$ with $u_0 \in V \oplus \mathcal{H}^1$ and $f \in H \oplus \mathcal{H}^1$, by the equalities \eqref{equal1}, \eqref{equal2} and \eqref{equal3} the simplified Bardina equation can be written as
\begin{align}\label{BarEqu1}
\begin{cases}
v_t + \nu \curl\curl_n v + \mathrm{grad}\dfrac{u^2}{2} - u\times \curl_n u  + \nabla p = f,\\
\nabla \cdot v= \nabla \cdot u =0,\\
v= u - \alpha^2 \Delta u,\\
u(0)=u_0.
\end{cases}
\end{align}
Recall that $\mathbb{P}$ is an orthogonal projection on $H$ namely Helmholtz projection. Denote by $\mathbb{Q}$ the projection of $L^2(M)$ on the space of harmonic forms $\mathcal{H}^1$. Putting
$$f=f_1+f_2, \, u(t)=u_1(t)+u_2(t), \, f_1,u_1(t)\in V, \, f_2,u_2(t)\in \mathcal{H}^1,$$
$$u_0=u_{10}+u_{20},\, u_{10}=\mathbb{P}(u_0)\in V, \, u_{20}=\mathbb{Q}(u_0)\in \mathcal{H}^1.$$
By applying the projection $\mathbb{P}+\mathbb{Q}$ on the simplified Bardina equation \eqref{BarEqu1}, we get
\begin{equation}\label{eq1}
\frac{d}{dt}(u_1+\alpha^2Au_1) + \nu Au_1 + \mathbb{P}(\curl_n u_1\times u_1 + \curl_n u_1\times u_2) = f_1,
\end{equation}
\begin{equation}\label{eq2}
\frac{d}{dt}u_2 + \mathbb{Q}(\curl_n u_1\times u_2) = f_2
\end{equation}
In order to Equations \eqref{eq1} and \eqref{eq2} become dissipative, some dissipative term must be added to these equations for example $\sigma u$. Therefore, we obtain
\begin{equation}\label{eq1'}
\frac{d}{dt}(u_1+\alpha^2Au_1) + \nu A(u_1 + \alpha^2Au_1) + \mathbb{P}(\curl_n u_1\times u_1 + \curl_n u_1\times u_2) + \sigma u_1 = f_1,
\end{equation}
\begin{equation}\label{eq2'}
\frac{d}{dt}u_2 + \mathbb{Q}(\curl_n u_1\times u_2) + \sigma u_2 = f_2.
\end{equation} 
These equations can be expressed in the simple form as
\begin{equation}\label{eq3}
\frac{d}{dt} (u+\alpha^2 Au) + \nu A(u + \alpha^2Au) + \mathbb{B}(u,u) + \sigma u = f,
\end{equation}
or 
$$\frac{d}{dt}v + \nu Av + \mathbb{B}(u,u) + \sigma u = f,$$
where
$$\mathbb{B}(u,u) = (\mathbb{P}+\mathbb{Q})(\curl_nu \times u).$$
\begin{definition}
Let $f \in H\oplus \cal{H}^1$ and $u_0\in V \oplus \cal{H}^1$ and $T>0$. A weak solution of Equation \eqref{eq3} is $u = u_1 + u_2: \, u_1 \in L^2([0,T],D(A))\cap C([0,T],V)$ and $u_2 \in C^1([0,T], \cal{H}^1)$  with $\dfrac{du_1}{dt}\in L^2([0,T],H)$ and such that for each $\omega = \omega_1 + \omega_2: \, \omega_1\in D(A)$ and $\omega_2 \in \cal{H}^1$: 
\begin{equation}\label{WeakSol}
\partial_t \left<v,w\right> + \nu \left<\curl_n v, \curl_n w\right> + b(u,u,\omega) + \sigma \left< u,\omega\right> = \left<f,\omega \right>,
\end{equation}
where $b(u,u,\omega)= \int_M \left<\mathbb{B}(u,u),\omega \right> dM = \int_{M} \curl_n u \times u \cdot \omega dM$. Equation \eqref{WeakSol} can be understood in the sense that for $t_0,\, t\in [0,T]$, we have the intergral equation
\begin{eqnarray*}
\left< v(t),\omega \right> - \left< v(t_0),\omega \right> + \nu\int_{t_0}^t\left< v(s),A\omega\right> ds 
&+& \int_{t_0}^t\left< \mathbb{B}(u(s),u(s)), \omega\right>ds \cr
&+&  \sigma \int_{t_0}^t\left< u(s), \omega\right>ds= \int_{t_0}^t\left< f,\omega\right>ds.
\end{eqnarray*}
\end{definition}
The bilinear operator $b(u,u,\omega)$ is generalized by trilinear form $b(u,v,\omega)$ on $H^1(TM)^3$ in the following formula
\begin{eqnarray*}
b(u,v,\omega) &=& \int_{M}\nabla_uv \cdot \omega dM = \int_M u^k\nabla_k v^i \omega^j g_{ij}dM \cr
&=& \frac{1}{2} \int_M \left( -u\times v \cdot \curl_n\omega + \curl_n u\times v\cdot\omega - u\times \curl_n v \cdot\omega \right) dM,
\end{eqnarray*}
where $u,v,\omega\in H^1(TM)$.
\begin{lemma}\label{basicEstimates}
The trilinear for $b(u,v,\omega)$ has the following properties (see \cite{Il1990,Il1988,Il1994})
\begin{itemize}
\item[i)] $|b(u,v,\omega)| \leq c\left\|u \right\|_1\left\|v \right\|_1\left\|\omega \right\|_1$.
\item [ii)] $|b(u,u,v)| \leq c'|u| \left\| u \right\|_1 \left\| v\right\|_1.$
\item[iii)] If $\dive u=0$ then $b(u,v,v)=0$, $b(u,v,\omega)=-b(u,\omega,v)$ and $b(u,u,Au) = 0$.
\end{itemize}
\end{lemma}

\section{Solvability and the existence of global attractor}\label{wellposed}
\subsection{The existence and uniqueness of the weak solutions}
We state and prove the existence and uniqueness of the weak solution of Equation \eqref{eq3} in the following theorem
\begin{theorem}\label{existence}
Let $u_0\in V\oplus \mathcal{H}^1$ and $f\in H\oplus \mathcal{H}^1$, then the equations \eqref{eq1'} and \eqref{eq2'} i.e equation \eqref{eq3} posseses a unique weak solution $u=u_1+u_2$: $u_1\in L^2([0,T],D(A))\cap C([0,T],V)$ and $u_2\in C^1([0,T],\mathcal{H}^1)$.
\end{theorem}
\begin{proof}  
The proof of the theorem is according the Galerkin scheme and then using Aubin's lemma. Recall that the orthonormal basis of $H$ is $\left\{\omega_i\right\}_1^\infty$ i.e the eigenfunctions of the Stokes operator $A=\curl\curl_n$. Let $\left\{ h_j\right\}_1^n$ be an orthonormal basis of the space of harmonic forms $\cal{H}^1$. Then we obtain that $\left\{ \omega_i\oplus h_j \right\}_{i=1,j=1}^{i=\infty,j=n}:=\left\{ \zeta_m \right\}_1^\infty$ is an orthonormal basis of $H\oplus \cal{H}^1$. The finite dimensional Galerkin approximation, based on this basis to the equation \eqref{eq3} is
\begin{align}\label{Galerkin}
\begin{cases}
\frac{d}{dt}(u_m+\alpha^2Au_m)+ \nu A(u_m+\alpha^2Au_m) + P_m\mathbb{B}(u_m,u_m)+ \sigma u_m = P_mf,\\
u_m(0)=P_mu(0),
\end{cases}
\end{align}
where $u_m:=P_mu= u_{m1}+u_{m2}$ and $P_mf=f_{m1}+f_{m2}$.\\

{\bf Step 1. $H^1$-estimates.} Taking the scalar product in $L^2(TM)$ of \eqref{Galerkin} and $u_m$, we get
\begin{eqnarray*}
&&\frac{1}{2}\frac{d}{dt} (|u_m|^2 + \alpha^2 \left\| u_{m1} \right\|^2) + \nu (\left\| u_{m1} \right\|^2 + \alpha^2| Au_{m1} |^2) + \sigma (|u_{m1}|^2 + |u_{m2}|^2) \cr
&=& |\left<P_mf,u_m \right>|= |\left<f,u_m \right>|\cr
&\leq& |\left<f_1,u_{m1}\right>|_{H} + |\left<f_2,u_{m2} \right>|_{\mathcal{H}^1}.
\end{eqnarray*}
By Cauchy-Schwarz inequality, we have
\begin{equation*}
|\left< f_1,u_{m1}\right>| \leq |A^{-1}f_1||Au_{m1}|,\, |A^{-1/2}f_1|\left\|u_{m1} \right\|
\end{equation*}
and by Young's inequality we have
\begin{equation*}
|\left< f_1,u_{m1}\right>|\leq \frac{|A^{-1}f_1|^2}{2\nu\alpha^2} + \frac{\nu}{2}\alpha^2|Au_{m1}|^2,\, \frac{|A^{-1/2}f_1|^2}{2\nu} + \frac{\nu}{2}\left\|u_{m1} \right\|^2.
\end{equation*}
And we have clearly
$$|\left< f_2,u_{m2} \right>|_{\mathcal{H}^1} \leq \frac{1}{2} \left( \frac{|f_2|^2}{\sigma} + \sigma|u_{m2}|^2 \right).$$
By putting $L_1 = \min\left\{\dfrac{|A^{-1}f_1|^2}{\nu\alpha^2},\dfrac{|A^{-1/2}f_1|^2}{\nu} , \dfrac{|f_2|^2}{\sigma}\right\}$, and by using the above inequalities we obtain that
\begin{equation}\label{ine1}
\frac{d}{dt}(|u_m|^2 + \alpha^2\left\| u_{m1}\right\|^2) + \nu(\left\|u_{m1} \right\|^2 + \alpha^2| Au_{m1}|^2) + \sigma( 2|u_{m1}|^2+|u_{m2}|^2) \leq L_1.
\end{equation}
Combining $|u_{m1}|\leq \lambda_1^{-1/2}\left\| u_{m1} \right\|$ and $|\curl_nu_{m1}| \leq \lambda_1^{-1/2}|Au_{m1}|$, we get
\begin{equation*}
\frac{d}{dt}(|u_m|^2 + \alpha^2\left\| u_{m1}\right\|^2) + \nu\lambda_1(|u_{m1}|^2 + \alpha^2\left\| u_{m1} \right\|^2) + \sigma(2|u_{m1}|^2+|u_{m2}|^2)\leq L_1,
\end{equation*}
which gives (since $Au_{m1}=Au_m$)
\begin{equation*}
\frac{d}{dt}(|u_m|^2 + \alpha^2\left\| u_m\right\|^2) + \delta(|u_m|^2 + \alpha^2\left\| u_m \right\|^2) \leq L_1,
\end{equation*}
where $\delta = \min\left\{\nu\lambda_1,\sigma\right\}.$ Using Gronwall's inequality we obtain that
\begin{equation}\label{abso1}
|u_m|^2 + \alpha^2\left\| u_m \right\|^2 \leq e^{-\delta t}(|u_{m0}|^2 + \alpha^2\left\|u_{m0} \right\|^2) + \frac{L_1}{\delta}(1-e^{-\delta t}) \leq |u_{m0}|^2 + \alpha^2\left\|u_{m0} \right\|^2 + \frac{L_1}{\delta}:=l_1.
\end{equation}
Therefore, for $0<T<+\infty$ and $u_{m0}:=u_m(0) \in V \oplus \cal{H}^1$, we have $u_m \in L^\infty([0,T],V \oplus \cal{H}^1)$ where the bound is uniform in $m$.

{\bf Step 2. $H^2$-estimates.} Integrating inequality \eqref{ine1} over $(t,t+r)$, we get
\begin{align}\label{ine2}
\nu\int_t^{t+r} (\left\| u_{m1}(s)\right\|^2 + \alpha^2|Au_{m1}(s)|^2)ds &\leq rL_1 + |u_m(t)|^2 + \alpha^2\left\| u_{m1}(t) \right\|^2 \nonumber\\
&\leq rL_1 + l_1.
\end{align}
Taking now the inner product of the Galerkin approximation \eqref{Galerkin} with $Au_{m}=Au_{m1}$, and note that (see \cite{Il1994} Lemma 3.1)
\begin{equation}\label{ProCurl}
\left< \curl_n u_{m1}\times u_{m1},Au_{m1} \right>=\left< \curl_n u_{m1} \times u_2,Au_{m1} \right>=0,
\end{equation}
we get 
\begin{equation*}
\frac{1}{2}\frac{d}{dt}(\left\| u_{m1}\right\|^2 + \alpha^2|Au_{m1}|^2) + \nu (|Au_{m1}|^2 + \alpha^2 |A^{3/2}u_{m1}|^2) + \sigma \left\| u_{m1} \right\|^2 \leq |\left< f_1,Au_{m1}\right>|.
\end{equation*}
Observe that
\begin{equation*}
|\left< f_1,Au_{m1} \right>| \leq |A^{-1/2}f_1||A^{3/2}u_{m1}|,\, |f_1||Au_{m1}|.
\end{equation*}
Using again Young's inequality we have
\begin{equation*}
|\left<f_1,Au_{m1} \right>| \leq \frac{|A^{-1/2}f_1|^2}{2\nu\alpha^2} + \frac{\nu}{2}\alpha^2 |A^{3/2}u_{m1}|^2, \frac{|f_1|^2}{2\nu} + \frac{\nu}{2}|Au_{m1}|^2.
\end{equation*}
Putting $L_2 = \min\left\{ \dfrac{|A^{-1/2}f_1|^2}{\nu \alpha^2},\dfrac{|f_1|^2}{\nu} \right\}$, then we have
\begin{equation*}
\frac{d}{dt}\left(\left\|u_{m1} \right\|^2 + \alpha^2|Au_{m1}|^2\right) + \nu \left( |Au_{m1}|^2 + \alpha^2|A^{3/2}u_{m1}|^2 \right) + 2\sigma \left\| u_{m1}\right\|^2 \leq L_2.
\end{equation*}
Combining with $|Au_{m1}| \leq \lambda_1^{-1}|A^{3/2}u_{m1}|^2$ (Poincar\'e inequality), we get
\begin{equation*}
\frac{d}{dt}\left(\left\|u_{m1} \right\|^2 + \alpha^2|Au_{m1}|^2\right) + \delta'( \left\| u_{m1} \right\|^2 + \alpha^2 |Au_{m1}|^2) \leq L_2.
\end{equation*}
where $\delta'=\min\left\{\nu\lambda_1,2\sigma \right\}$. Hence
\begin{equation}\label{ine3}
\frac{d}{dt}\left(\left\|u_{m1} \right\|^2 + \alpha^2|Au_{m1}|^2\right) \leq L_2.
\end{equation}
Integrating the above inequality over $(s,t)$ to obtain that
\begin{equation}\label{ine4'}
\left\|u_{m1}(t) \right\|^2 + \alpha^2|Au_{m1}(t)|^2 \leq \left\|u_{m1} (s)\right\|^2 + \alpha^2|Au_{m1}(s)|^2 + (t-s)L_2,
\end{equation}
continuting integrating over $(0,t)$ and using \eqref{ine2} we obtain that
\begin{equation}\label{ine4}
t\left(\left\|u_{m1}(t) \right\|^2 + \alpha^2|Au_{m1}(t)|^2 \right) \leq \frac{1}{\nu}(tL_1+l_1) + \frac{t^2}{2}L_2,
\end{equation}
for all $t>0$. For $t\geq \dfrac{1}{\nu\lambda_1}$, we integrate \eqref{ine3} over $\left(t-\dfrac{1}{\nu\lambda_1},t \right)$ to establish
\begin{equation}\label{ine5}
\frac{1}{\nu\lambda_1}\left(\left\|u_{m1}(t) \right\|^2 + \alpha^2|Au_{m1}(t)|^2 \right) \leq \frac{1}{\nu}\left( \frac{1}{\nu\lambda_1}L_1 + l_1\right) + L_2\left( \frac{1}{2\nu\lambda_1}\right)^2.
\end{equation}
The inequalities \eqref{ine4} and \eqref{ine5} yield that there exists a function $l_2(t)$ satisfying the following conditions
\begin{itemize}
\item[i)] For all $t>0$ then $l_2(t)<+\infty$ and $\lim_{t\rightarrow +\infty}l_2(t)<+\infty$,
\item[ii)] If $u_{10} \in V$ but $u_{10} \notin D(A)$, then $\lim_{t\rightarrow 0^+}l_2(t)=+\infty$.
\end{itemize}
and
\begin{equation*}
\left\| u_{m1}(t)\right\|^2 + \alpha^2|Au_{m1}(t)|^2 \leq l_2(t),\, t>0.
\end{equation*}
\begin{remark}
Inequality \eqref{ine4'} yields that if $u_{10}:=u_1(0)\in D(A)$, then $u_{m1}(.)$ is bounded uniformly in $L^\infty([0,T],D(A))$ independently of $m$. On the other hand, if $u_{10}\in V$ but $u_{10} \notin D(A)$, then $u_{m1}\in L^\infty_{loc}((0,T],D(A))\cap L^2([0,T],D(A))$.
\end{remark}

{\bf Step 3. Estimates for $\dfrac{dv_{m}}{dt}$ and $\dfrac{du_{m}}{dt}$.} For each $\omega = \omega_1+\omega_2$ where $\omega_1\in D(A)$ and $\omega_2\in \cal{H}^1$, we have
\begin{equation*}
\frac{d}{dt}\left< v_{m},\omega\right> = -\nu \left<Av_{m1},\omega_1 \right> - \left< P_m\mathbb{B}(u_m,u_m)),\omega \right> - \sigma\left< u_{m},\omega \right> + \left< P_mf,\omega\right>.
\end{equation*}
Since $u_{m1}$ is uniformly bounded with respect to $m$ in $L^2([0,T],D(A))$, hence $v_{m1}$ is uniformly bounded in $L^2([0,T],H)$, as a consequence $Av_{m1}$ is uniformly bounded in $L^2([0,T],D(A)')$. Now we observe that
\begin{eqnarray*}
|\left< P_mf,\omega \right>| = |\left< f, P_m\omega \right>| &\leq& |\left< f_1, \omega_1 \right>|_H + |\left< f_2, \omega_2 \right>|_{\cal{H}^1} \cr
&\leq& |A^{-1}f_1||A\omega_1| + |f_2||\omega_2| \leq \lambda_1^{-1}|f_1||A\omega_1| + |f_2||\omega_2|.
\end{eqnarray*}
By $ii)$ of Lemma \ref{basicEstimates}
\begin{eqnarray*}
|\left<P_m\mathbb{B}(u_m,u_m),\omega\right>| &\leq&  |b(u_m,u_m,\omega)| \leq c'|u_m| \norm{u_m}_1 \norm{\omega}_1\cr
&\leq& c'|u_m| (\norm{u_{m1}} + |u_{m2}|)(\norm{\omega_1}+|\omega_2|)\cr 
&\leq& c'|u_m| (\norm{u_{m1}} + |u_{m2}|) \left(\lambda_1^{-1/2}|A\omega_1| + |\omega_2| \right).
\end{eqnarray*}
Moreover
\begin{eqnarray*}
|\left< u_{m},\omega \right>| &\leq& |\left< u_{m1},\omega_1 \right>| + |\left< u_{m2},\omega_2 \right>|\cr
&\leq& |A^{-1}u_{m1}||A\omega_1| + |u_{m2}||\omega_2| \cr
&\leq& \lambda_1^{-1}|u_{m1}||A\omega_1| + |u_{m2}||\omega_2|.
\end{eqnarray*}
We therefore conclude $\left\|\dfrac{dv_{m}}{dt}\right\|_{L^2([0,T],(D(A)\oplus \cal{H}^1)')}$ and $\left\|\dfrac{du_{m}}{dt}\right\|_{L^2([0,T],H\oplus \cal{H}^1)}$ are uniformly bounded with respect to $m$. By Aubin compactness theorem, there is a subsequene $u_{m'}(t)$ and a function $u(t)$ such that
\begin{eqnarray*}
&&u_{m'}(t)\longrightarrow u(t) \, \text{weakly in} \, L^2([0,T],D(A)\oplus \cal{H}^1),\cr
&&u_{m'}(t)\longrightarrow u(t) \, \text{strongly in} \, L^2([0,T],V \oplus \cal{H}^1),\cr
&&u_{m'}\longrightarrow u \, \text{in} \, C([0,T],H\oplus \cal{H}^1).
\end{eqnarray*}
These are equivalent to
\begin{eqnarray*}
&&v_{m'}(t)\longrightarrow v_1(t) \, \text{weakly in} \, L^2([0,T],H \oplus \cal{H}^1),\cr
&&v_{m'}(t)\longrightarrow v_1(t) \, \text{strongly in} \, L^2([0,T],(V \oplus \cal{H}^1)'),\cr
&&v_{m'}\longrightarrow v_1 \, \text{in} \, C([0,T],(D(A)\oplus \cal{H}^1)').
\end{eqnarray*}
Now relabel $u_{m'}$ (resp. $v_{m'}$) by $u_{m}$ (resp. $v_{m}$). For $\omega=\omega_1+\omega_2$ where $\omega_1\in D(A)$ and $\omega_2\in \cal{H}^1$, we have
\begin{eqnarray*}
\left< v_{m}(t),\omega \right> &+& \nu\int_{t_0}^t\left< v_{m1}(s),A\omega_1\right> ds + \int_{t_0}^t\left< \mathbb{B}(u_m(s),u_m(s)),P_m\omega\right>ds \cr
&&+ \sigma\int_{t_0}^t\left< u_{m}(s),P_m\omega\right>ds =\left< v_m(t_0),\omega\right> + \left< f,P_m\omega\right>(t-t_0),
\end{eqnarray*}
for all $t_0,t\in [0,T]$. Since the sequence $v_{m}(t)$ converges weakly in $L^2([0,T],H\oplus \cal{H}^1)$, $v_{m1}(t)$ converges weakly in $L^2([0,T],H)$ then
\begin{equation*}
\lim_{m\rightarrow\infty}\int_{t_0}^t \left< v_{m1}(s),A\omega_1\right>ds = \int_{t_0}^t\left< v_1(s),A\omega_1\right>ds,
\end{equation*}
and there is a subsequence of $v_{m}$ and relabel by $v_{m}$ which converges almost everywhere on $[0,T]$ to $v(t)$ in $(H\oplus \cal{H}^1)'\simeq H \oplus \cal{H}^1$. Therefore
\begin{eqnarray*}
&&\left< v_{m}(t),\omega \right> \longrightarrow \left< v(t),\omega\right>,\cr
&&\left< v_{m}(t_0),\omega \right> \longrightarrow \left< v(t_0),\omega\right>,
\end{eqnarray*}
almost everywhere for $t,t_0 \in [0,T]$.

Now we treat the convergence of the nonlinear term $\int_{t_0}^t \left< \mathbb{B}(u_m(s),u_m(s)),P_m\omega \right>ds$. We have
\begin{eqnarray*}
&&\left| \int_{t_0}^t \left< \mathbb{B}(u_m(s),u_m(s)),P_m\omega\right> - \left< \mathbb{B}(u(s),u(s)),\omega\right> ds \right|\cr
&\leq& \left| \int_{t_0}^t \left< \mathbb{B}(u_m(s),u_m(s)),P_m\omega-\omega \right> ds \right|:=I_m^I\cr
&&+ \left| \int_{t_0}^t \left< \mathbb{B}(u_m(s)-u(s),u_m(s)),\omega \right> \right|:=I_m^{II}\cr
&&+ \left| \int_{t_0}^t \left< \mathbb{B}(u(s),u_m(s)-u(s)),\omega \right> \right|:=I_m^{III}.
\end{eqnarray*}
To estimate $I_m^I$, we observe that there exists a constant $c''>0$ such that 
$$| \left< \mathbb{B}(u_m(s),u_m(s)),P_m\omega-\omega \right> | \leq c''\norm{u_{m}(s)}_1 \left\| P_m\omega-\omega \right\|_{L^\infty(TM)}|u_m(s)|.$$
Applying Agmon inequality in $2$-dimension: $\left\| \omega \right\|_{L^\infty(TM)} \leq C |\omega|^{1/2}_{L^2(TM)}\left\| \omega\right\|^{1/2}_{H^2(TM)}$, we get
$$| \left< B(u_m(s),u_m(s)),P_m\omega-\omega \right> | \leq c''C|u_m(s)|\norm{u_{m}(s)}_1 |P_m\omega-\omega|^{1/2} (| A(P_m\omega_1-\omega_1 )|+|P_m\omega_2-\omega_2| )^{1/2}.$$
Therefore
$$I_m^I \leq c''C\left( \int_{t_0}^t |u_m(s)|^2ds \right)^{1/2}\left( \int_{t_0}^t \left\|u_{m}(s)\right\|_1^2ds \right)^{1/2}|P_m\omega-\omega|^{1/2} (| A(P_m\omega_1-\omega_1 )|+ |P_m\omega_2-\omega_2|)^{1/2}. $$
Since $u_{m}$ is uniformly bounded in $L^\infty([0,T],V\oplus \cal{H}^1)$ and $u_m$ is uniformly bounded in $L^\infty([0,T],H \oplus \cal{H}^1)$ independently of $m$ (Step 1), we obtain that $\lim_{m\rightarrow\infty}I_m^I=0$.

Similarly, Agmon inequality and Poincar\'e inequality yeild 
\begin{eqnarray*}
\left\| \omega \right\|_{L^\infty(TM)} &\leq  C | \omega |^{1/2}_{L^2(TM)}\left\|\omega\right\|^{1/2}_{H^2(TM)} \leq C( |\omega_1| + |\omega_2|)^{1/2} (|A\omega_1|+|\omega_2|)^{1/2}\cr
&\leq C ( \lambda_1^{-1}|A\omega_1| + |\omega_2|)^{1/2} (|A\omega_1|+|\omega_2|)^{1/2},
\end{eqnarray*}
then $I_m^{II}$ can be estimated as
$$I_m^{II} \leq C \left( \int_{t_0}^t |u_m(s)-u(s)|^2ds\right)^{1/2} \left( \int_{t_0}^t \left\| u_{m}(s) \right\|_1^2\right)^{1/2} ( \lambda_1^{-1}|A\omega_1| + |\omega_2|)^{1/2} (|A\omega_1|+|\omega_2|)^{1/2}.$$
Combining with $u_m \rightarrow u$ strongly in $L^2([0,T],V\oplus \cal{H}^1)$ and the boundedness of $\norm{u_{m}}_1$, we get $\lim_{t\rightarrow\infty}I_m^{II}=0$. By the same manner we also have $\lim_{t\rightarrow\infty}I_m^{III}=0$. Therefore,
$$\int_{t_0}^t \left< \mathbb{B}(u_m(s),u_m(s)),P_m\omega \right>ds \longrightarrow \int_{t_0}^t \left< \mathbb{B}(u(s),u(s)),\omega \right>ds \, \text{as} \, m \rightarrow \infty.$$
We conclude that for almost everywhere $t_0,t\in [0,T]$ and every $\omega\in D(A) \oplus \cal{H}^1$: 
\begin{eqnarray*}
\left< v(t),\omega \right> - \left< v(t_0),\omega \right> + \nu\int_{t_0}^t\left< v_1(s),A\omega_1\right> ds &+& \int_{t_0}^t\left< \mathbb{B}(u(s),u(s)), \omega\right>ds + \sigma\int_{t_0}^t\left< u(s),\omega\right>ds \cr
&=&\int_{t_0}^t\left< f,\omega\right>ds.
\end{eqnarray*}
On the other hand, $v_1\in L^2([0,T],H)$ and $\omega_1\in D(A)$, then we have
$$ \left| \int_{t_0}^t \left< v_1(s),A\omega_1\right>ds \right| \leq \left( \int_{t_0}^t|v_1(s)|^2 ds\right)^{1/2} \left( \int_{t_0}^t|A\omega_1|^2 ds\right)^{1/2} \rightarrow 0  \hbox{  as  } t \rightarrow t_0.$$
And since $u\in L^\infty([0,T],V \oplus \cal{H}^1)$, then
\begin{eqnarray*}
&&\left|\int_{t_0}^t \left< \mathbb{B}(u(s),u(s)),\omega \right>ds  \right| \leq \cr
&&|\omega|_{L^\infty(TM)} \left( \int_{t_0}^t|u(s)|^2 ds\right)^{1/2}\left( \int_{t_0}^t \left\|u_(s)\right\|_1^2 ds \right)^{1/2} \longrightarrow 0 \hbox{  as  }  t\rightarrow t_0.
\end{eqnarray*}
Therefore for $t_0,t\in [0,T]$, $\left< v(t),\omega\right>\longrightarrow \left< v(t_0),\omega\right>$ as $t\rightarrow t_0$ for every $\omega\in D(A) \oplus \cal{H}^1$. Since $D(A) \oplus \cal{H}^1$ is dense in $V \oplus \cal{H}^1$ we have $\left< v(t),\omega\right>\longrightarrow \left< v(t_0),\omega\right>$ as $t\rightarrow t_0$ for every $\omega\in V \oplus \cal{H}^1$ hence $v\in C([0,T], (V\oplus \cal{H}^1)')$ and $u\in C([0,T],V \oplus \cal{H}^1)$. The existence of the solution $u$ for equation \eqref{eq3} holds. Finally, $u_2 \in C^1([0,T],\cal{H}^1)$ is a general property of solutions of linear finite-dimensional systems of differential
equations.

{\bf Step 4. Uniqueness.} Now we prove the uniqueness of the solution of Equation \eqref{eq3}. Suppose that $\omega=\omega_1+\omega_2$ is another solution of \eqref{eq3}. Putting $z=u-\omega$, hence $z_0=0$ and 
\begin{equation*}
\frac{d}{dt} (z+\omega+\alpha^2 A(z_1+\omega_1)) + \nu A(z_1+\omega_1 + \alpha^2A(z_1+\omega_1)) + \mathbb{B}(z+\omega,z+\omega) + \sigma (z+\omega) = f.
\end{equation*}
Subtracting this equation with 
\begin{equation*}
\frac{d}{dt} (\omega+\alpha^2 A\omega_1) + \nu A(\omega_1 + \alpha^2A\omega_1) + \mathbb{B}(\omega,\omega) + \sigma \omega = f,
\end{equation*}
we obtain the variation form
\begin{equation*}
\frac{d}{dt} (z+\alpha^2 Az_1) + \nu A(z_1 + \alpha^2Az_1) + \mathbb{B}(\omega,z) + \mathbb{B}(z,\omega) + \sigma z = 0.
\end{equation*}
Taking the scalar product in $L^2(TM)$ of the above equation and $z$
\begin{equation*}
\frac{d}{dt}(|z|^2+ \alpha^2\left\| z_1\right\|^2) + 2\nu(\left\| z_1\right\|^2 + \alpha^2|Az_1|^2) + 2b(z,\omega,z) + \sigma |z|^2=0.
\end{equation*}
Using $ii)$ in Lemma \ref{basicEstimates}, we get
\begin{equation*}
\frac{d}{dt}(|z|^2+ \alpha^2\left\| z_1\right\|^2) + 2\nu(\left\| z_1\right\|^2 + \alpha^2|Az_1|^2) + \sigma |z|^2 =2b(z,z,\omega) \leq 2c'|z|\left\|z\right\|_1\left\|\omega\right\|_1.
\end{equation*}
Putting $z(t)=e^{\nu t}\tilde{z}(t)$, we obtain that
\begin{equation*}
\frac{d}{dt}(|\tilde{z}|^2 + \alpha^2\left\| \tilde{z}_1 \right\|^2) + 2\nu(|\tilde{z}|^2 + \left\| \tilde{z}_1\right\|^2 + \alpha^2\left\| \tilde{z}_1 \right\|^2 + \alpha^2|A\tilde{z}_1|^2) + \sigma|\tilde{z}|^2 \leq 2c'|\tilde{z}|\left\| \tilde{z}\right\|_1 \left\| \omega\right\|_1
\end{equation*}
hence
\begin{equation*}
\frac{d}{dt}(|\tilde{z}|^2 + \alpha^2\left\| \tilde{z}_1 \right\|^2) + 2\nu\left\|\tilde{z}\right\|_1^2  \leq 2\nu\left\| \tilde{z}\right\|_1^2 + \frac{2c'}{\nu}|\tilde{z}|^2\left\| \omega\right\|_1^2
\end{equation*}
which implies
\begin{equation*}
\frac{d}{dt}(|\tilde{z}|^2 + \alpha^2\left\| \tilde{z}_1 \right\|^2) \leq \frac{2c'}{\nu}(|\tilde{z}|^2+\alpha^2 \left\| \tilde{z}_1 \right\|^2 )\left\| \omega\right\|_1^2.
\end{equation*}
Using Gronwall inequality, we can establish that
\begin{equation*}
|\tilde{z}(t)|^2 + \alpha^2\left\| \tilde{z}_1(t) \right\|^2 \leq \left(|\tilde{z}_0|^2 + \alpha^2\left\| \tilde{z}_{10} \right\|^2\right)\exp\left( \int_0^t \frac{2c'}{\nu} \left\| \omega(s)\right\|_1^2 ds\right).
\end{equation*}
Since $\tilde{z}_0=0$, we obtain that $\tilde{z} = 0$. The proof of uniqueness is completed. 
\end{proof}

\subsection{The existence of global attractor}
We recall the $H^1$-estimates which are obtained in the previous section
\begin{equation}\label{abso1}
|u(t)|^2 + \alpha^2\left\| u(t)\right\|^2 \leq e^{-\delta t}(|u_0|^2 + \alpha^2\left\|u_0 \right\|^2) + \frac{L_1}{\delta}(1-e^{-\delta t}).
\end{equation}
Hence
\begin{equation*}
\limsup_{t\rightarrow\infty}(|u(t)|^2+\alpha^2\| u(t)\|^2) \leq 2\frac{L_1}{\delta}:=\rho_0^2,
\end{equation*}
where $L_1= \min\left\{\dfrac{|A^{-1}f_1|^2}{\nu\alpha^2},\dfrac{|A^{-1/2}f_1|^2}{\nu} , \dfrac{|f_2|^2}{\sigma}\right\}$.

Now we have the $H^2$-estimates as
\begin{equation}\label{abso2}
\left\| u_1(t)\right\|^2 + \alpha^2|Au_1|^2 \leq e^{-\delta't}(\left\| u_{10}\right\|^2 + \alpha^2|Au_{10}|^2) + \frac{L_2}{\delta'}(1-e^{-\delta t}).
\end{equation}
Hence
\begin{equation*}
\limsup_{t\rightarrow\infty}(\|u(t)\|^2+\alpha^2 |A u(t)|^2) \leq \frac{L_2}{\delta'}:=\rho_1^2.
\end{equation*}

If the space $V\oplus \mathcal{H}^1$ is equipped with the following scalar product
\begin{eqnarray}\label{normVH}
 [u,v]_{V\oplus \cal{H}^1} &=& \left< \curl_n u, \curl_n v\right> +  \left< u, v\right>\cr
 &=& \left< u,(A+I)v \right>,
\end{eqnarray}
then after long enough time, $u(t)$ enters a ball in $V\oplus \cal{H}^1$ with the radius squared: $\rho^2=\rho_0^2+ \rho_1^2$. This means that the semigroup $S_t$ generated by \eqref{eq3} acts on $V\oplus \mathcal{H}^1$, it has an absorbing ball $B_{V\oplus\mathcal{H}^1}(0) \subset V\oplus \mathcal{H}^1$ with the radius $\rho$. The existence of absorbing ball $B_{D(A)\oplus \cal{H}^1}(0)$ in $D(A)\oplus\cal{H}^1$ is done in the same manner.

Now following Rellich lemma $S_t: V\oplus \mathcal{H}^1 \longrightarrow D(A)\oplus \mathcal{H}^1 \Subset V\oplus \mathcal{H}^1$, for $t>0$, is a compact semigroup from $V$ into itself. Since $S(t)B_{V\oplus \mathcal{H}^1}(0) \subset B_{V\oplus \mathcal{H}^1}(0)$, then the set $C_s:= \overline{\cup_{t\geq s}S(t)B_{V\oplus\mathcal{H}^1}(0)}^{V\oplus\mathcal{H}^1}$ is nonempty and compact in $V\oplus \mathcal{H}^1$. By the monotonic property of $C_s$ for $s>0$ and by the finite intersection property of compact sets, the set
$$\mathcal{A} = \cap_{s>0}C_s \subset V\oplus\mathcal{H}^1$$
is a nonempty compact set, and also the unique global attractor in $V\oplus \mathcal{H}^1$.

\section{Dimensions of global attractor}
\subsection{Fundamental theorem}
Let $H$ be an Hilbert space, $X$ be a compact set in $H$ and $S_t$ the nonlinear continuous semigroup generated by the evolution equation
$$\partial_tu = F(u), \, u(0)=u_0,$$
and suppose that
$$S_tX=X \hbox{  for  } t\geq 0.$$
The Hausdorff and fractal dimensions of $X$ are estimated by using the uniform Lyapunov exponents (see Theorem 3.3 in \cite{Te1984} for the origin case: $S_t$ is uniformly differentiable). The result was extended to the case of a uniformly quasi-differentiable semigroup in \cite{Il2001,Il2004}. 
\begin{definition}
The semigroup $S_t$ is uniformly quasi-differentiable on $X$ for each $t$ if for all $u,\, v \in X$ there exists a linear operator $DS_t(u)$ such that
$$\norm{S_t(u)-S_t(v)-DS_t(u)(u-v)} \leq h(r)\norm{u-v},$$
where $\norm{u-v}\leq r$, $h(r)\rightarrow 0$ as $r\rightarrow 0$ and $\sup_{t\in [0,\, 1]} \sup_{u\in X}\norm{DS_t(u)}_{\mathscr{L}(H,H)}<\infty$.
\end{definition}
The following result is establised in \cite{Il2001} (see Theorem 2.1).
\begin{theorem}\label{TheoremDim}
We assume that the mapping $u \rightarrow S_tu_0$ is uniformly quasi-differentiable in $H$ and its quasi-differential is a linear operator $L(t,u_0):\zeta\in H \rightarrow U(t)\in H$, where $U(t)$ is the solution of the first variation equation
\begin{equation}\label{HF1}
\partial_t U = \mathscr{L}(t,u_0)U, \, \, U(0)=\zeta.
\end{equation}
We assume, in addition, that for a fixed $t$ the operator $L(t, u_0) = DS_t(u)$ is compact and norm-continuous with respect
to $u \in X$.

For $N\geq 1, \, n\in \mathbb{N},$ we define $q_N$ by
\begin{equation}\label{Lyapunov}
q_N = \limsup_{t\rightarrow\infty}\sup_{u_0\in X}\sup_{\zeta_i\in H,\norm{\zeta_i}\leq 1, i=1,...,N}\left(  \frac{1}{t}\int_0^t\mathrm{Tr}\mathscr{L}(\tau, u_0)\circ Q_N(\tau)d\tau\right),
\end{equation}
where $Q_N(\tau)$ is the orthogonal projection in $H$ into $\mathrm{Span}\left\{U^1(\tau)...U^N(\tau)\right\}$, and $U^i(t)$ is the solution of \eqref{HF1} with $U^i(0)=\zeta_i$. 

Suppose $q_N \leq f(N)$, where $f$ is concave. The Hausdorff and fractal dimensions of $X$ have the same upper bound
$$\dim_H X \leq \dim_F X\leq N_*,$$
where $N_*\geq 1$ is such that $f(N_*) = 0$.
\end{theorem}
The concave condition of $f$ can be replaced by the condition that the quasi-differential $DS_t(u)$ contracts $N_*$-dimensional volumes uniformly for $u \in X$ (see Theorem 2.1 \cite{Il2004}).


\subsection{Estimate of the attractor's dimensions}\label{Dim}

\subsubsection{Upper bound}
For simplicity we consider the simplified Bardina equation on the $2$-sphere $S^2$ which is a specific case of $M$ with $\mathcal{H}^1= \left\{ \vec{0} \right\}$.
First, we rewrite the equation to a vorticity scalar form. Recall that the origin equation is
\begin{equation*}
(u_t - \alpha^2\Delta u_t) - \nu (\Delta u - \alpha^2 \Delta^2 u) + (u\cdot \nabla) u + \nabla p = f.
\end{equation*}
Let $u = -\curl\psi$. Then applying $\curl_n$ to the above equation we obtain
\begin{equation*}
(\Delta\psi_t - \alpha^2\Delta^2\psi_t) - \nu\Delta(\Delta\psi - \alpha^2\Delta^2\psi) + (u\cdot\nabla)\Delta\psi = \curl_n f,
\end{equation*}
where $u = n \times \nabla \psi$.

Putting $\varphi = \curl_n u = \Delta \psi$ we get
\begin{equation}\label{CurlEq1}
(\varphi_t - \alpha^2 \Delta \varphi_t) - \nu\Delta(\varphi - \alpha^2 \Delta \varphi) + u \cdot \nabla \varphi = \curl_n f. 
\end{equation}
Hence
\begin{equation}\label{CurlEq11}
\varphi_t - \nu \Delta \varphi + (I-\alpha^2\Delta)^{-1}(u \cdot \nabla \varphi) = (I-\alpha^2\Delta)^{-1}\curl_n f. 
\end{equation}
We define the bilinear operator $J(a,b)$ as follows
$$J(a,b) = n\times \nabla a\cdot\nabla b.$$
We have
$$u \cdot \nabla \varphi = J(\psi,\varphi) = J(\Delta^{-1}\varphi,\varphi)$$
Therefore, Equation \eqref{CurlEq11} becomes
\begin{equation}\label{CurlEq2}
\varphi_t - \nu \Delta \varphi + (I-\alpha^2\Delta)^{-1}J(\Delta^{-1}\varphi,\varphi) = (I-\alpha^2\Delta)^{-1}\curl_n f. 
\end{equation}
\begin{remark}
The bilinear operator $J(a,b)$ has the following properties
$$\int_{S^2} J(a,b)dx = \int_{S^2} J(a,b)bdx=0 \hbox{  and  } \int_{S^2} J(a,b)cdx = \int_{S^2} J(b,c)adx.$$
\end{remark}
By multiplying \eqref{CurlEq1} by $\varphi$ in $L^2(S^2)$ and by using \eqref{ProCurl} we obtain that
\begin{equation*}
\frac{1}{2}\frac{d}{dt} (|\varphi|^2+ \alpha^2|\nabla\varphi|^2) + \nu (|\nabla\varphi|^2 + \alpha^2|\Delta\varphi|^2) = \left<\curl_n f_1,\varphi \right>= \left< f, \curl_n\varphi \right>.
\end{equation*}
Therefore,
\begin{equation*}
\frac{d}{dt} (|\varphi|^2+ \alpha^2|\nabla\varphi|^2) + 2\nu (|\nabla\varphi|^2 + \alpha^2|\Delta\varphi|^2) \leqslant \frac{|f|^2}{\nu} + \nu|\nabla\varphi|^2.
\end{equation*}
Using the Poincar\'e and Gronwall inequalities and integrating with respect to $t$ yield
\begin{equation}\label{Evar}
\limsup_{t\to\infty}|\varphi(t)|^2 \leqslant \frac{|f|^2}{\lambda_1\nu^2}
\end{equation}
and
\begin{equation}\label{Evarphi}
\limsup_{t\to\infty}\frac{1}{t}\int_0^t|\nabla\varphi(\tau)|^2d\tau \leqslant \frac{|f|^2}{\nu^2}.
\end{equation}

We consider the variational equation corresponding to \eqref{CurlEq2}:
\begin{equation*}
\Phi_t = \nu\Delta\Phi - (I-\alpha^2\Delta)^{-1}J(\Delta^{-1}\Phi,\varphi) - (I-\alpha^2\Delta)^{-1}J(\Delta^{-1}\varphi,\Phi):= \mathscr{L}(t,\varphi_0)\Phi,
\end{equation*}
where $\Phi(0) = \zeta$.

It is standard to show that this equation has a unique solution denoted by
$$L(t, \varphi(0))\zeta := \Phi(t).$$
Using the general theorems in \cite{BaVi,Te1984} we can show that the semigroup $S_t$ is uniformly quasi-differentiable on the attractor $\mathcal{A}$ of the simplified Bardina equation.

We now estimate the fractal dimension of the attractor.
\begin{theorem}\label{HF}
The Hausdorff and fractal dimension of the attractor $\mathcal{A}$ of the simplified Bardina equation are finite and satisfy 
\begin{equation}\label{Upper1}
\dim_H\mathcal{A} \leqslant \dim_F \mathcal{A} \leqslant G^{2/3} \left(\frac{(4+\epsilon_G)^3}{3\pi(1+\alpha^2)^3}(\log G - \frac{1}{2}\log \frac{\pi}{2})\right)^{1/3},
\end{equation}
\begin{equation}\label{Upper2}
\dim_H\mathcal{A} \leqslant \dim_F \mathcal{A} \leqslant \left( \frac{12}{\sqrt{\pi(1+\alpha^2)^3}} \right)^{2/3} G^{2/3} \left( \log G + \frac{1}{2} + \log \frac{3\sqrt{2}}{\sqrt{\pi (1+\alpha^2)^3}})\right)^{1/3},
\end{equation}
where $G= \dfrac{|f|}{\nu^2\lambda_1}$ is the Grashof number and $\epsilon_G \to 0$, when $G\to \infty$. 
\end{theorem}
\begin{proof}
Let
$$\mathbb{H} = L^2(S^2)\cap \left\{ \varphi: \int_{S^2} \varphi dS^2 =0 \right\} \hbox{  and  } \mathbb{H}^1 = H^1(S^2) \cap \mathbb{H}.$$
We define a scalar product on $\mathbb{H}^1$ depending $\alpha$ by
\begin{equation}\label{NormVH}
\left<\left< \varphi, \varphi'\right>\right>_\alpha= \left< \varphi, (I-\alpha^2\Delta)\varphi'\right>.
\end{equation}
Clearly, we have that
$$|\varphi|^2 = \norm{\varphi}_{\alpha}^2 - \alpha^2 |\nabla\varphi|^2 \leqslant \norm{\varphi}^2_{\alpha} - \alpha^2\lambda_1 |\varphi|^2.$$
Hence
$$|\varphi|^2 \leqslant \frac{1}{1+\alpha^2\lambda_1}\norm{\varphi}^2_{\alpha}.$$
In the space $Q_N(\tau)(\mathbb{H})$ we take an orthonormal basis $\left\{\theta_i \right\}_{i=1}^N \subset \mathbb{H}^1$ with respect to \eqref{NormVH}. 

Now we have
\begin{eqnarray}\label{OperatorTrace1}
&&\mathrm{Tr}\mathscr{L}(\tau,\varphi_0)\circ Q_N(\tau) = \sum_{i=1}^N \left<\left<\mathscr{L}(\tau,\varphi_0)\theta_i, \theta_i \right>\right> \cr
&=& -\nu\sum_{i=1}^N \left<\left< \Delta\theta_i,\theta_i \right>\right> \cr
&&- \sum_{i=1}^N \left<\left<(I+\alpha^2A)^{-1}(J(\Delta^{-1}\theta_i,\varphi)+ J(\Delta^{-1}\varphi,\theta_i) ), \theta_i \right>\right> \cr
&=& -\nu\sum_{i=1}^N(|\nabla\theta_i|^2 + \alpha^2|\Delta\theta_i|^2) - \sum_{i=1}^N \left<J(\Delta^{-1}\theta_i,\varphi)+ J(\Delta^{-1}\varphi,\theta_i), \theta_i \right> \cr
&=& -\nu\sum_{i=1}^N(|\nabla\theta_i|^2 + \alpha^2|\Delta\theta_i|^2) - \sum_{i=1}^N \left<J(\Delta^{-1}\theta_i,\varphi), \theta_i \right>\cr
&\leq& -\nu\sum_{i=1}^N(|\nabla\theta_i|^2 + \alpha^2|\Delta\theta_i|^2) - \int_M \sum_{i=1}^N \theta_i (n\times\nabla\Delta^{-1}\theta_i)\cdot \nabla\varphi dx\cr
&\leq& -\nu\sum_{i=1}^N(|\nabla\theta_i|^2 + \alpha^2|\Delta\theta_i|^2) + \int_M \left(\sum_{i=1}^N \theta^2_i \right)^{1/2} \left(\sum_{i=1}^N |v_i|^2 \right)^{1/2}|\nabla\varphi| dx\cr
&\leq& -\nu\sum_{i=1}^N(|\nabla\theta_i|^2 + \alpha^2|\Delta\theta_i|^2) + \norm{\rho}^{1/2}_\infty \left(\sum_{i=1}^N |\theta_i|^2 \right)^{1/2}|\nabla\varphi|\cr
&\leq& -\nu\sum_{i=1}^N(|\nabla\theta_i|^2 + \alpha^2|\Delta\theta_i|^2) + \frac{1}{1+\alpha^2\lambda_1}\norm{\rho}^{1/2}_\infty \left(\sum_{i=1}^N \norm{\theta_i}_{\alpha}^2 \right)^{1/2}|\nabla\varphi|\cr
&\leq& -\nu\sum_{i=1}^N(|\nabla\theta_i|^2 + \alpha^2|\Delta\theta_i|^2) + \frac{1}{1+\alpha^2\lambda_1}\norm{\rho}^{1/2}_\infty N^{1/2}|\nabla\varphi|,
\end{eqnarray}
where
$$\rho(s) = \sum_{i=1}^N|v_i(s)|^2 = \sum_{i=1}^n |n \times \nabla\Delta^{-1}\theta_i|^2.$$
With the scalar product \eqref{NormVH} the following estimate of the function $\rho$ is valid (for details see Appendix).
\begin{eqnarray}\label{INE}
&&2\sqrt{\pi(1+\alpha^2\lambda_1)}\norm{\rho}_\infty^{1/2} \leqslant (2\log(k+1)+1)^{1/2} + \sqrt{2}(k+1)^{-1}\left( \lambda_1^{-1}\sum_{i=1}^N|\nabla\theta_i|^2 \right)^{1/2} \cr
&\leqslant& (2\log(k+1)+1)^{1/2} + \sqrt{2}(k+1)^{-1}\left( \lambda_1^{-1}\sum_{i=1}^N(|\nabla\theta_i|^2 + \alpha^2|\Delta\theta_i|^2) \right)^{1/2},
\end{eqnarray}
where $k$ is a positive integer.

Since on the $S^2$ the eigenvalues of $\Delta$ are $\lambda_n = n(n+1)$ of multiplicity $2n+1$ for $n=1,2,...$, we have
$$T(t,\varphi_0):= \sum_{i=1}^N(|\nabla\theta_i|^2 + \alpha^2|\Delta\theta_i|^2) \geqslant \sum_{i=1}^N  \lambda_i \geqslant \frac{\lambda_1}{4} N^2.$$
Hence
$$N \leqslant 2((\lambda_1)^{-1}T)^{1/2}.$$
Equation \eqref{OperatorTrace1} implies now,
\begin{eqnarray*}
&&\mathrm{Tr}\mathscr{L}(\tau,\varphi_0)\circ Q_N(\tau) \leqslant -\nu\lambda_1 (\lambda_1^{-1} T) \cr
&&+ \pi^{-1/2}(1+\alpha^2\lambda_1)^{-3/2} \left( (2\log(k+1)+1)^{1/2} + \sqrt{2}(k+1)^{-1}(\lambda_1^{-1}T) \right)(\lambda_1^{-1}T)^{1/4}|\nabla\varphi|^2
\end{eqnarray*}
If we take $k = [\lambda_1^{-1}T]-1$, then 
$$ (2\log(k+1)+1)^{1/2} + \sqrt{2}(k+1)^{-1}(\lambda_1^{-1}T) \leqslant C_N(\log(\lambda_1^{-1}T)+1)^{1/2},$$
where
\begin{align*}
c_N = \begin{cases}
3 \hspace*{3,2cm} N \geqslant 1,\cr
\sqrt{2} + \epsilon_N \hspace*{2cm} \epsilon_N\to 0, \hbox{  when  } N \to \infty.
\end{cases}
\end{align*}
Putting
$$\mathcal{N}(t,\varphi_0)^2: = \frac{1}{t}\int_0^t\lambda_1^{-1}T(\tau,\varphi_0)d\tau.$$
Since $N \leqslant 2((\lambda_1)^{-1}T)^{1/2}$, we have $\mathcal{N}\geqslant \frac{N}{2}$.

Therefore, we have
\begin{eqnarray*}
&&\frac{1}{t}\mathrm{Tr}\mathscr{L}(\tau,\varphi_0)\circ Q_N(\tau)\cr
&\leqslant& -\nu\lambda_1\mathcal{N}^2 + \pi^{-1/2}(1+\alpha^2\lambda_1)^{-3/2}c_N\frac{1}{t}\int_0^t(\log(\lambda_1^{-1}T)+1)^{1/2}(\lambda_1^{-1}T)^{1/4}|\nabla\varphi(\tau)|d\tau\cr
&\leqslant& -\nu\lambda_1\mathcal{N}^2 + \pi^{-1/2}(1+\alpha^2\lambda_1)^{-3/2}c_N \left(\frac{1}{t}\int_0^t(\log(\lambda_1^{-1}T)+1)(\lambda_1^{-1}T)^{1/2}d\tau\right)^{1/2}\cr
&&\times \left(\frac{1}{t}\int_0^t |\nabla\varphi(\tau)|^2d\tau \right)^{1/2}\cr
&\leqslant& -\nu\lambda_1\mathcal{N}^2 + \pi^{-1/2}(1+\alpha^2\lambda_1)^{-3/2}c_N(\log \mathcal{N}^2+ 1)^{1/2}\mathcal{N}^{1/2}\frac{|f|}{\nu}\cr
&=& \nu\lambda_1\mathcal{N}^2 (-\mathcal{N}^{3/2} + K(2\log \mathcal{N}+1)^{1/2}):= g(\mathcal{N}), 
\end{eqnarray*}
where $K= \pi^{-1/2}(1+\alpha^2\lambda_1)^{-3/2}c_N G, \, G= \frac{|f|}{\lambda_1\nu^2}$ and on the last inequality we have used the inequality \eqref{Evarphi} and applied Jensen's inequality to the concave
function $x\mapsto x^{1/2}(1+\log x)$, with $x= \lambda_1^{-1}T\geqslant 1$.

We have $g(\mathcal{N}) \geqslant 0$ if
$$\mathcal{N}^{3/2} \leqslant K(2\log \mathcal{N}+1)^{1/2}.$$
It is equivalent to
\begin{equation}\label{log}
 3\log\mathcal{N} - \log(2\log\mathcal{N}+1) \leqslant 2\log K.
\end{equation}
Using a fact that 
\begin{equation}\label{log'}
\log(2\log\mathcal{N}+1) < \epsilon_K \log\mathcal{N}, \hbox{ where } \epsilon_K \to 0, \, \mathcal{N}\to \infty,
\end{equation}
we obtain since \eqref{log} that
$$\mathcal{N} \leqslant K^{2/3}\left( \frac{4}{3} + \epsilon_K \right)(\log K)^{1/3}.$$
Therefore, we can replace the function $g$ by a concave function $g'$ such that $g'(\mathcal{N})\geqslant g(\mathcal{N})$ and $g'(\mathcal{N}) =0$ if 
$$\mathcal{N} = K^{2/3}\left( \frac{4}{3} + \epsilon_K \right)(\log K)^{1/3}.$$
By using Theorem \ref{TheoremDim} and $N<2\mathcal{N}$ we obtain that
\begin{equation*}
\dim_H\mathcal{A} \leqslant \dim_F \mathcal{A} \leqslant G^{2/3} \left(\frac{(4+\epsilon_G)^3}{3\pi(1+\alpha^2\lambda_1)^3}(\log G - \frac{1}{2}\log \frac{\pi}{2})\right)^{1/3},
\end{equation*}
where $G= \dfrac{|f|}{\nu^2\lambda_1}$ is the Grashof number and $\epsilon_G \to 0$, when $G\to \infty$.
The upper bounds \eqref{Upper1} holds.

If we replace \eqref{log'} by
$$\log(2\log\mathcal{N}+1)<\log 2 + \log \mathcal{N}, \, \mathcal{N}\geqslant 1,$$
then by the same way as above we can also obtain that
\begin{equation*}
\dim_H\mathcal{A} \leqslant \dim_F \mathcal{A} \leqslant \left( \frac{12}{\sqrt{\pi(1+\alpha^2\lambda_1)^3}} \right)^{2/3} G^{2/3} \left( \log G + \frac{1}{2} + \log \frac{3\sqrt{2}}{\sqrt{\pi (1+\alpha^2\lambda_1)^3}})\right)^{1/3}.
\end{equation*}
The upper bound \eqref{Upper2} holds.

\end{proof}
\begin{remark}\label{Remark}
\item[(i)] As $\alpha$ tends to zero we get the same upper bound of the Haussdorff and fractal's dimensions of the global attractor for the Navier-Stokes equation on $S^2$.
\item[(ii)] On the two dimensional closed manifold $M$ we can also prove an estimate that likes \eqref{INE} as (see Appendix)
\begin{equation}\label{INE'}
\sqrt{(1+\alpha^2\lambda_1)}\norm{\rho}_\infty^{1/2} \leqslant L \left( (2\log(k+1)+1)^{1/2} + \sqrt{\lambda_1}(k+1)^{-1}\left( \lambda_1^{-1}\sum_{i=1}^N|\nabla\theta_i|^2 \right)^{1/2} \right).
\end{equation}
Therefore, if $\mathcal{H}^1=\left\{ \vec{0} \right\}$ the same proof (without explicit constants, of course) gives the
estimate of the attractor dimension,
$$\dim_F \mathcal{A} \leqslant c(\frac{1}{1+\alpha^2\lambda_1})G^{2/3}(\log G +1)^{1/3}$$
for the simplified Bardina equation on a simply connected compact manifold or in a simply
connected bounded domain $\Omega$,supplemented with boundary conditions $u \cdot n|_{\partial \Omega}=0, \curl_n u|_{\partial \Omega} =0$.
\end{remark}
\begin{theorem}
For the multiply connected manifold or domain $M$, i.e, $\mathcal{H}^1 \ne \left\{ \vec{0} \right\}$. Assume that the phase space is assumed to be orthogonal to the finite dimensional space of
harmonic vector fields. The attractor's dimensions of the simplified Bardina on $M$ satisfy that
$$\dim_H\mathcal{A}_{V\oplus \mathcal{H}^1}\leqslant \dim_F\mathcal{A}_{V\oplus \mathcal{H}^1} \leqslant c(\frac{1}{1+\alpha^2\lambda_1})G^{2/3}(\log G +1)^{1/3} +n,$$
where $\dim \mathcal{H}^1=n$.
\end{theorem}
\begin{proof}
On $M$ we recall that the simplified Bardina equation is
\begin{equation*}
\frac{d}{dt}(u_1+\alpha^2Au_1) + \nu A(u_1 + \alpha^2Au_1) + \mathbb{P}(\curl_n u_1\times u_1 + \curl_n u_1\times u_2) + \sigma u_1 = f_1,
\end{equation*}
\begin{equation*}
\frac{d}{dt}u_2 + \mathbb{Q}(\curl_n u_1\times u_2) + \sigma u_2 = f_2.
\end{equation*}
Rewrite the first equation in the scalar form we get
\begin{equation*}
\frac{d}{dt}\varphi + \nu A\varphi + (I+\alpha^2A)^{-1}\mathbb{P}[(u_1+u_2)\cdot \nabla \varphi] + \sigma (I+\alpha^2A)^{-1}\varphi = (I+\alpha^2A)^{-1}\curl_n f_1,
\end{equation*}
\begin{equation*}
\frac{d}{dt}u_2 + \mathbb{Q}(\curl_n u_1\times u_2) + \sigma u_2 = f_2,
\end{equation*}
where $u=u_1+u_2$ and $\varphi = \curl_n u_1$.

The variational equations are
\begin{eqnarray*}
&&(\delta \varphi)_t = -\nu A \delta\varphi + (I+\alpha^2A)^{-1}J(A^{-1}\delta \varphi,\varphi) + (I+\alpha^2A)^{-1}J(A^{-1}\varphi,\delta\varphi)\cr
&& - (I+\alpha^2A)^{-1}\mathbb{P}(\delta u_2\cdot \nabla \varphi+ u_2\cdot\nabla\delta\varphi) - \sigma\delta\varphi,\cr
&&(\delta u_2)_t = -\mathbb{Q}(\varphi \times \delta u_2 + (\delta\varphi)\times u_2) - \sigma \delta u_2.
\end{eqnarray*}
In the matrix form, we put $U= \begin{bmatrix}
        \delta\varphi \\
        \delta u_2
    \end{bmatrix}$, then
\begin{align}\label{Matrix} 
U_t =  -\nu \begin{bmatrix}
        A & 0\\
        0 & 0
    \end{bmatrix}U -   \begin{bmatrix}
        A_{11} & A_{12}\\
        A_{21} & A_{22}
    \end{bmatrix}U,  
\end{align}
where
\begin{eqnarray*}
&&A_{11} = -(I+\alpha^2A)^{-1}J(A^{-1}*,\varphi) + (I+\alpha^2A)^{-1}J(A^{-1}\varphi,*) - (I+\alpha^2A)^{-1}\mathbb{P}(u_2\cdot\nabla *) -\sigma *\cr
&&A_{12} = (I+\alpha^2A)^{-1}\mathbb{P}(* \cdot \nabla \varphi)\cr
&&A_{21} = -\mathbb{Q}(*\times u_2), \, A_{22} = -\mathbb{Q}(\varphi\times *)-\sigma *.
\end{eqnarray*}

We define the scalar product 
$$\left<\left<(\varphi,u_2),(\varphi',u_2')  \right>\right> = \left< \varphi, (I+\alpha^2A)\varphi \right> + \left< u_2, u_2'\right>.$$
We take $\left\{ (\theta_i,0), (0,h_j) \right\},\, i=1,2...N,\, j=1,2...,l$ be an orthonormal basic in $\mathbb{H}\oplus \mathcal{H}^1$,
where $\left\{ \theta_i\right\}_1^N$ are orthonormal in $\mathbb{H}$ with respect to the norm $\norm{\theta}_\alpha = |\theta|+ \alpha^2|\nabla\theta|$ and $\left\{ h_i\right\}_1^l$ are orthonormal in $\mathcal{H}^1$.

Using the variational equation \eqref{Matrix} and
$$\left< \curl_n u_1 \times h_j, h_j \right> = 0,$$
we can see that the harmonic fields have no contribution on $q_{N+n}$, then
$$\dim_F\mathcal{A}_{V\oplus \mathcal{H}^1}(u) \leqslant \dim_F\mathcal{A}_{\mathbb{H}}(\varphi) + n.$$
By using (ii) Remark \ref{Remark} we have
$$\dim_F\mathcal{A}_{\mathbb{H}}(\varphi) \leqslant c(\frac{1}{1+\alpha^2\lambda_1})G^{2/3}(\log G +1)^{1/3}.$$
The proof is completed.
\end{proof}

\subsubsection{Lower bound}
In order to check the sharp upper bound of the attractor's dimensions we will consider its lower bound. 
We known that if the Grashof number $G$ is arbitrarily large, then the corresponding attractor of the Navier-Stokes equation on $S^2$ has dimension $0$ and reduces to the globally attractive stationary point (see the examples in \cite{Ma,Il1994}).  We give a similar argument for the simplified Bardina equation on $S^2$ with in the following proposition.
\begin{proposition}
If $G= \dfrac{|f|}{\lambda_1\nu}<\dfrac{1}{\sqrt 2}$ then $\dim\mathcal{A} = 0$.
\end{proposition}
\begin{proof}
Suppose that $\bar{\varphi}= \Delta\bar{\psi}$ is a stationary solution of \eqref{CurlEq1}. Taking the scalar product of \eqref{CurlEq1} with $\bar{\varphi}$ we get
\begin{equation*}
|\nabla\bar{\varphi}|^2 \leqslant \frac{|f|^2}{\nu^2}.
\end{equation*}
Let $\varphi = \bar{\varphi} + \varphi'$ be a solution of the evolution problem \eqref{CurlEq1}. Then $\varphi'$ satisfies the following equation
$$(\varphi'_t - \alpha^2\Delta \varphi'_t) - \nu(\Delta\varphi' - \alpha^2\Delta\varphi') + J(\Delta^{-1}\bar{\varphi},\varphi')+ J(\Delta^{-1}\varphi',\varphi') + J(\Delta^{-1}\varphi', \bar{\varphi}) =0.$$
Taking the scalar product with $\varphi'$ we obtain that
\begin{eqnarray}\label{E2}
&&\frac{1}{2}\frac{d}{dt}(|\varphi'|^2 + \nu|\nabla\varphi'|^2) + \nu(|\nabla\varphi'|^2 + \alpha^2|\Delta\varphi'|^2) = -\int_{S^2}J(\Delta^{-1}\varphi', \bar{\varphi})\varphi'dx\cr
&\leqslant& \int_{S^2}|\nabla\Delta^{-1}\varphi'||\nabla\bar{\varphi}||\varphi'|dx\cr
&\leqslant& \norm{\nabla\bar{\varphi}}_{L^2}\norm{\nabla\Delta^{-1}\varphi'}_{L^4}\norm{\varphi'}_{L^4}\cr
&\leqslant& \frac{|f|}{\nu}\norm{\nabla\Delta^{-1}\varphi'}_{L^4}\norm{\varphi'}_{L^4}
\end{eqnarray}
Denote $\norm{.}_{L^2} = |.|$, by H\"older's and Ladyzhenskaya’s inequality we have
$$\norm{\varphi'}_{L^4} \leqslant c_1|\varphi'||\nabla\varphi'|, \, \varphi'\in \mathbb{H}^1(S^2)\cap \mathbb{H},$$
$$\norm{\nabla\psi'}_{L^4} \leqslant c_2 |\nabla\psi'|^{1/2}|\Delta\psi'|^{1/2}, \, \psi'\in \mathbb{H}^2(S^2),$$
where $c_1,c_2\leqslant 2^{1/4}$.

Combining the above inequalities with \eqref{E2} yield
$$\frac{1}{2}\frac{d}{dt}(|\varphi'|^2 + \nu|\nabla\varphi'|^2) + \left( \nu - c_1c_2\frac{|f|}{\lambda_1\nu}  \right)|(\nabla\varphi'|^2 + \alpha^2|\Delta\varphi'|^2)  \leqslant 0.$$
Therefore, if $G<1/(c_1c_2)<1/\sqrt 2$, then the stationary solution $\bar{\varphi}$ is globally exponentially
attractive, and $\mathcal{A} = \bar{\varphi}$.
\end{proof}
Since a global attractor is a maximal strictly invariant compact set, it follows that the attractor contains the unstable manifolds of stationary points, that is the invariant manifolds along which the solutions convergence exponentially to the stationary points as $t$ tends to infinity.
Since that Liu \cite{Liu}, Ilyin and Titi \cite{Il2004'} provided lower bounds of the attractor's dimensions for the Navier-Stokes equation (the case $\alpha=0$) and the Navier-Stokes-alpha equation on the two-dimensional torus $T$ by constructing a family of stationary solutions arising from the family of Kolmogorov flows. In particular, they proved that
$$\dim\mathcal{A} \geqslant cG^{2/3}.$$
In the recent paper, Ilyin and Zelik \cite{Il2021} develop the methods in \cite{Liu,Il2004'} to establish the lower bound depending the damped coefficient for the attractor's dimensions of the damped $2$D Euler-Bardina equation on $T^2$.

In the next, we will develop the method in \cite{Liu,Il2004',Il2021} to establish the lower bound for the attractor's dimension of the simplified Bardina equation on $T^2=[0; \, 2\pi]\times[0; \, 2\pi]$. Recall that the scalar vorticity form of the equation is
\begin{equation*}
(\varphi_t - \alpha^2 \Delta \varphi_t) - \nu\Delta(\varphi - \alpha^2 \Delta \varphi) + J(\Delta^{-1}\varphi,\varphi) = \curl_n f. 
\end{equation*}
Putting $\psi = \varphi -  \alpha^2\Delta\varphi$, then
\begin{equation}\label{LowerEq}
\psi_t - \nu\Delta\psi + J((\Delta - \alpha^2\Delta^2)^{-1}\psi,(I-\alpha^2\Delta)^{-1}\psi) = \curl_n f.
\end{equation}

We consider the following family of forces depending on the integer parameter $s$:
\begin{align*}
f=f_s =\begin{cases}
f_1 = \frac{1}{\sqrt{2}\pi}\nu^2\lambda s^2\sin s x_2,\\
f_2 =0,
\end{cases}
\end{align*}
where we choose the parameter $\lambda:=\lambda(s)$ later. Then, we have
$$|f| = \nu^2\lambda s^2, \, G= \lambda s^2$$
and
\begin{equation}\label{flow}
\curl_n f_s = F_s = -\frac{1}{\sqrt{2}\pi}\nu^2\lambda s^3\cos sx_2, \, |\curl_nf| = \nu^2\lambda s^3.
\end{equation}
Corresponding to the family \eqref{flow} is the family of stationary solutions
\begin{equation*}
\psi_s = -\frac{1}{\sqrt{2}\pi}\nu\lambda s\cos sx_2
\end{equation*}
of Equation \eqref{LowerEq} due to $\psi_s$ depends only on $x_2$, the nonlinear term vanishes
$$J((\Delta - \alpha^2\Delta^2)^{-1}\psi_s,(I-\alpha^2\Delta)^{-1}\psi_s)=0$$
and the equality $-\nu \Delta\psi_s = F_s$ is verified directly.

We linearize \eqref{LowerEq} about the stationary solution \eqref{flow} and consider the eigenvalue
problem
\begin{eqnarray}\label{EP}
\mathcal{L}_s\psi :&=& J((\Delta - \alpha^2\Delta^2)^{-1}\psi_s,(I-\alpha^2\Delta)^{-1}\psi) \cr
&&+ J((\Delta - \alpha^2\Delta^2)^{-1}\psi,(I-\alpha^2\Delta)^{-1}\psi_s) - \nu\Delta\psi = -\sigma\psi.
\end{eqnarray}
We use the orthonormal basis of trigonometric functions, which are the eigenfunctions of the Laplacian on the two-dimensional torus,
$$\left\{ \frac{1}{\sqrt{2}\pi}\sin kx, \frac{1}{\sqrt{2}\pi}\cos kx  \right\}, \, kx = k_1x_1+k_2x_2,$$
$$k \in \mathbb{Z}^2_+ = \left\{ k\in \mathbb{Z}_0^2| k_1\geq 0,\, k_2\geq 0 \right\}\cup \left\{k\in \mathbb{Z}_0^2| k_1\geq 1, k_2\leq 0  \right\}$$
and we rewrite $\psi$ as a Fourier series
$$\psi = \frac{1}{\sqrt{2}\pi}\sum_{k\in \mathbb{Z}_+^2} a_k \cos kx + b_k \sin kx.$$
Plugging this into \eqref{EP} and using the fact that $J(a,b)=-J(b,a)$ we obtain that
\begin{eqnarray}\label{EP1}
\frac{\lambda s}{\sqrt{2}\pi (s^2+ \alpha^2s^4)}&&\sum_{k\in \mathbb{Z}_+^2} \left( \frac{k^2-s^2}{k^2+\alpha^2k^4} \right)J(\cos sx_2, a_k \cos kx + b_k \sin kx)+\cr
&&+ \sum_{k\in \mathbb{Z}_+^2}(k^2+ \hat{\sigma})(a_k\cos kx + b_k \sin kx)=0,
\end{eqnarray}
where $\hat{\sigma}=\sigma/\nu$.

We can calculate that
\begin{eqnarray*}
J(\cos s x_2, \cos(k_1x_1+k_2x_2)) &=& -k_1s \sin sx_2 \sin(k_1x_1 + k_2x_2)\cr
&=& \frac{k_1 s}{2} (\cos (k_1x_1 +(k_2+s)x_2)) -\cos (k_1x_1+(k_2-s)x_2)
\end{eqnarray*}
and
\begin{eqnarray*}
J(\cos s x_2, \sin(k_1x_1+k_2x_2)) &=& k_1s \sin sx_2 \cos(k_1x_1 + k_2x_2)\cr
&=& \frac{k_1 s}{2} (\sin (k_1x_1 +(k_2+s)x_2)) - \sin (k_1x_1+(k_2-s)x_2).
\end{eqnarray*}
Substituting these equalities into \eqref{EP1} and regroup the terms with $\cos(k_1x_1+k_2x_2)$, we get the following equation for the coefficients $a_{k_1,k_2}$ 
\begin{eqnarray*}
&&-\Lambda(s)k_1 \left( \frac{k_1^2+(k_2+s)^2-s^2}{k_1^2+(k_2+s)^2 + \alpha^2(k_1^2 + (k_2+s)^2)^2} \right) a_{k_1k_2+s}\cr
&&+\Lambda(s)k_1 \left( \frac{k_1^2+(k_2-s)^2-s^2}{k_1^2+(k_2-s)^2 + \alpha^2(k_1^2 + (k_2-s)^2)^2} \right) a_{k_1k_2-s} + (k^2+\hat{\sigma})a_{k_1k_2} =0,
\end{eqnarray*}
where 
\begin{equation}
\Lambda = \Lambda(s):= \frac{s^2\lambda}{2\sqrt{2}\pi (s^2+ \alpha^2s^4)} = \frac{\lambda }{2\sqrt{2}\pi(1+\alpha^2 s^2)}.
\end{equation}
Similarly the equation for $b_{k_1,k_2}$ has also this form.

We put
$$a_{k_1k_2} \left( \frac{k^2-s^2}{k^2+ \alpha^2k^4} \right) =: c_{k_1k_2}.$$
and
$$k_1 = t, \, k_2= sn+r, \hbox{  and  } c_{t \, sn+r}= e_n,$$
$$t=1,2,..., \, r \in \mathbb{Z}, \, r_{\min} < r < r_{\max},$$
where the numbers $r_{\min}$ and $r_{\max}$ satisfy that $r_{\max} - r_{\min} <s$ and will be specified below we obtain for each $t$ and $r$ the following three term recurrence relation:
\begin{equation}\label{d1}
d_ne_n + e_{n-1} - e_{n+1} =0 , \, n=0,\pm 1,\pm 2,...,
\end{equation}
where
\begin{equation}\label{d2}
d_n = \frac{(t^2+(sn+r)^2+ \alpha^2(t^2+(sn+r)^2)^2)(t^2+(sn+r)^2+\hat{\sigma})}{\Lambda t(t^2+(sn+r)^2-s^2)}.
\end{equation}
We look for non-trivial decaying solutions $\left\{ e_n \right\}$ of \eqref{d1} and \eqref{d2}. Each nontrivial decaying solution with $\mathrm{Re}(\hat{\sigma})>0$ produces an unstable eigenfunction $\psi$ of the eigenvalue problem \eqref{EP}.  
\begin{theorem}
Given an integer $s>0$ let a pair of integers $t,\, r$ belong to a bounded region $A(\delta)$ given by
\begin{gather}\label{Con}
t^2+r^2<s^2/3, \, t^2+(-s+r)^2>s^2, \, t^2+(s+r)^2>s^2, \, t\geqslant \delta s,\cr
r_{\min}<r<r_{\max}, \, r_{\min} = -s/6, \, r_{\max}=s/6, \, 0<\delta<1/\sqrt{3}.
\end{gather}
For any $\Lambda>0$ there exists a unique real eigenvalue $\hat{\sigma} = \hat{\sigma}(\Lambda)$, which increases monotonically as $\Lambda\to \infty$ and satisfies the following inequality
\begin{equation}\label{Estimate}
c_1(\alpha,t,r,s)\Lambda < \hat{\sigma} < c_2(\alpha,t,r,s)\Lambda. 
\end{equation}
The unique $\Lambda_0 = \Lambda_0(s)$ solving the equation
$$\hat{\sigma}(\Lambda_0) = 0$$
satisfes the two-sided estimate
\begin{gather}\label{LU}
\frac{1}{\sqrt 2}\delta^2s(1+\alpha^2s^2) < \Lambda < \frac{55\sqrt 5}{63\sqrt 2}\frac{s(1+\alpha^2s^2)}{\delta^2} \hbox{    for    } \alpha \geqslant 0,\cr
\frac{1}{\sqrt 2}\delta^2s < \Lambda < \frac{5}{3\sqrt 3}\frac{s}{\delta^2} \hbox{    for    } \alpha =0.
\end{gather}
\end{theorem}
\begin{proof}
We observe that the following inequalities hold for any $(t,\, r)$ satisfying \eqref{Con}: 
\begin{eqnarray}
&&s^2\leqslant t^2+(-s+r)^2 = \mathrm{dist}((0,s),(t,r))^2 \leqslant\mathrm{dist}((0,s),C)^2 = (5/3) s^2\cr
&&s^2\leqslant t^2+(s+r)^2 = \mathrm{dist}((0,-s),(t,r))^2 \leqslant\mathrm{dist}((0,-s),B)^2 = (5/3) s^2,
\end{eqnarray}
where $B = (\sqrt{11}s/6,\, s/6)$ and $C= (\sqrt{11}s/6, \, -s/6)$.

In view of \eqref{Con} for any real $\hat{\sigma}$ satisfying $\hat{\sigma}>-t^2-r^2$ we have in the recurrence
relation \eqref{d1} and \eqref{d2} as
\begin{equation}\label{Con1}
d_n >0 \hbox{   for   } n \neq 0 \hbox{   and   } \lim_{|n|\to \infty}d_n = \infty.
\end{equation}
The main tool in the analysis of \eqref{d1} are continued fractions and a variant
of Pincherle's theorem saying that under condition \eqref{Con1} the recurrence relation \eqref{d1} has a decaying solution
$\left\{ e_n \right\}$ with $\lim_{|n|\to\infty}e_n = 0$ if and only if
\begin{equation}
-d_0 = \frac{1}{d_{-1}+\frac{1}{d_{-2}+...}}+ \frac{1}{d_{1}+\frac{1}{d_{2}+...}}.
\end{equation}
Now, we set
\begin{equation}\label{f}
f(\hat{\sigma}) = -d_0 = \frac{(t^2+r^2+ \alpha^2(t^2+r^2)^2)(t^2+r^2+\hat{\sigma})}{\Lambda t(s^2- t^2 - r^2)},
\end{equation}
\begin{equation}\label{g}
g(\hat{\sigma}) = \frac{1}{d_{-1}+\frac{1}{d_{-2}+...}}+ \frac{1}{d_{1}+\frac{1}{d_{2}+...}}.
\end{equation}
The equation \eqref{f} leads to
$$f(-t^2-r^2) = 0 \hbox{  and  } f(\hat\sigma)\to 0 \hbox{  as  } \hat{\sigma}\to \infty.$$
Combining \eqref{g} and \eqref{d2} we have 
$$g(\hat\sigma)< \frac{1}{d_{-1}} + \frac{1}{d_1} \hbox{  and  } g(\hat\sigma)\to 0 \hbox{  as  }\hat\sigma\to \infty.$$
Therefore, there exists a $\hat{\sigma}> -t^2-r^2$ such that
\begin{equation}\label{eq}
f(\hat\sigma) = g(\hat\sigma).
\end{equation}
From elementary properties of continued fractions we deduce as in \cite{Liu} that the $\hat{\sigma}$ so obtained is unique and increases monotonically with $\Lambda$.

To establish \eqref{Estimate} we deduce from \eqref{g} and \eqref{eq} that
\begin{equation}\label{two}
\frac{1}{d_{-1}+\frac{1}{d_{-2}}}+ \frac{1}{d_{1}+\frac{1}{d_{2}}} < f(\hat\sigma )< \frac{1}{d_{-1}}+ \frac{1}{d_{1}}.
\end{equation}
Using the conditions $t^2+(-s+r)^2>s^2$ and $t^2+(s+r)^2>s^2$, we deduce from \eqref{d2} that
\begin{eqnarray}
\frac{1}{d_{\pm 1}} &=& \frac{\Lambda t}{t^2+(s\pm r)^2+\hat{\sigma}}\frac{t^2+ (s\pm r)^2 - s^2}{t^2 + (s\pm r)^2 + \alpha^2(t^2 + (s\pm r)^2)^2} \cr
&\leqslant& \frac{\Lambda t}{s^2+\hat{\sigma}}\frac{1}{1+ \alpha^2(t^2+(s\pm r)^2)}\leqslant \frac{\Lambda t}{s^2+\hat{\sigma}}\frac{1}{1+ \alpha^2 s^2}.
\end{eqnarray}
Therefore, from the right-hand inequality in \eqref{two} it follows that
\begin{equation*}
f(\hat{\sigma}) = \frac{(t^2+r^2+ \alpha^2(t^2+r^2)^2)(t^2+r^2+\hat{\sigma})}{\Lambda t(s^2- t^2 - r^2)}< \frac{1}{d_{-1}} + \frac{1}{d_1} < \frac{2\Lambda t}{s^2+\hat{\sigma}}\frac{1}{1+ \alpha^2 s^2}.
\end{equation*}
Hence
\begin{eqnarray}\label{upper}
(t^2+r^2+\hat{\sigma})(s^2+\hat{\sigma}) &<& \frac{2\Lambda^2 t^2(s^2-(t^2+r^2))}{(t^2+r^2+\alpha^2(t^2+r^2)^2)(1+\alpha^2s^2)}\cr
&\leqslant& \frac{2\Lambda^2 t^2s^2}{(t^2+\alpha^2t^4)(1+\alpha^2s^2)} \leqslant \frac{2\Lambda^2\delta^{-2}s^2}{(1+\alpha^2s^2)^2},
\end{eqnarray}
which gives the right-hand side inequality in \eqref{Estimate}:
\begin{equation*}
\hat{\sigma} \leqslant c_2(\alpha,t,r,s)\Lambda \hbox{   as   } \Lambda\to \infty.
\end{equation*}
Morefuther, we set $\hat{\sigma}=0$ in \eqref{upper} and use the condition $t^2+r^2>\delta^2s^2, \, 0<\delta<1$ we get
$$\delta^2s^4 < \frac{2\Lambda^2\delta^{-2}s^2}{(1+\alpha^2s^2)^2}.$$
Therefore, we obtain the lower bound of $\Lambda$ as in the left-hand side of \eqref{LU}:
$$\frac{1}{\sqrt 2}\delta^2s(1+\alpha^2s^2)<\Lambda.$$ 

From the left-hand side inequality in \eqref{two}, where $d_{-1}, \, d_1, \, d_{-2},\, d_2, \, f>0$, we see that
\begin{equation}\label{fd0}
fd_1 + \frac{f}{d_2}>1 \hbox{   and   } fd_{-1} + \frac{f}{d_{-2}}>1.
\end{equation}
We have
\begin{eqnarray}\label{fd1}
fd_1 &=& \frac{(t^2+r^2+ \alpha^2(t^2+r^2)^2)(t^2+r^2+\hat{\sigma})}{\Lambda t(s^2- t^2 - r^2)}\cr
&&\times \frac{t^2+(s+ r)^2+\hat{\sigma}}{\Lambda t}\frac{t^2 + (s+ r)^2 + \alpha^2(t^2 + (s+ r)^2)^2}{t^2+ (s+ r)^2 - s^2} \cr
fd_{-1} &=& \frac{(t^2+r^2+ \alpha^2(t^2+r^2)^2)(t^2+r^2+\hat{\sigma})}{\Lambda t(s^2- t^2 - r^2)}\cr
&&\times \frac{t^2+(s- r)^2+\hat{\sigma}}{\Lambda t}\frac{t^2 + (s- r)^2 + \alpha^2(t^2 + (s- r)^2)^2}{t^2+ (s- r)^2 - s^2}
\end{eqnarray}
and
\begin{eqnarray}\label{fd2}
\frac{f}{d_2} &=& \frac{(t^2+r^2+\hat{\sigma})}{t^2+(2s+ r)^2+\hat{\sigma}}\frac{(t^2+r^2+ \alpha^2(t^2+r^2)^2)}{(s^2- t^2 - r^2)}\cr
&&\times \frac{t^2+ (2s+ r)^2 - s^2}{t^2 + (2s+ r)^2 + \alpha^2(t^2 + (2s+ r)^2)^2} \cr
\frac{f}{d_{-2}} &=& \frac{(t^2+r^2+\hat{\sigma})}{t^2+(2s- r)^2+\hat{\sigma}}\frac{(t^2+r^2+ \alpha^2(t^2+r^2)^2)}{(s^2- t^2 - r^2)}\cr
&&\times \frac{t^2+ (2s- r)^2 - s^2}{t^2 + (2s- r)^2 + \alpha^2(t^2 + (2s- r)^2)^2}
\end{eqnarray}
The first factors in \eqref{fd2} are clearly less than one. It follows from \eqref{Con} that $4sr\leqslant 2s^2/3$ and $|2s+r|\geqslant 11s/6$. Therefore, we can control the right-hand sides of \eqref{fd2} as 
\begin{eqnarray*}
\frac{f}{d_2} &=& \frac{(t^2+r^2+ \alpha^2(t^2+r^2)^2)}{(s^2- t^2 - r^2)} \frac{(t^2+ (2s+ r)^2 - s^2)}{(t^2 + (2s+ r)^2 + \alpha^2(t^2 + (2s+ r)^2)^2)} \cr
&<& \frac{(s^2/3+ \alpha^2s^4/9)4s^2}{2s^2/3((11/6)^2s^2+ \alpha^2(11/6)^4)s^4} = \frac{2(1+\alpha^2s^2/3)}{(11/6)^2+\alpha^2s^2(11/6)^4} \leqslant \frac{72}{121}.
\end{eqnarray*}
We would like to remark that if $\alpha=0$ we can improve this estimate by
\begin{eqnarray*}
\frac{f}{d_2} &=& \frac{(t^2+r^2)}{(s^2- t^2 - r^2)} \frac{(t^2+ (2s+ r)^2 - s^2)}{(t^2 + (2s+ r)^2} < \frac{t^2+r^2}{s^2-t^2-r^2} \cr
&<& \frac{s^2/3}{2s^2/3}=\frac{1}{2}.
\end{eqnarray*}
Along with \eqref{fd0} we have that $fd_1>49/121$ for $\alpha\geqslant 0$, which for $r\geq 0$ gives that
\begin{eqnarray}\label{U}
\frac{49}{121}&<& fd_1 = \frac{(t^2+r^2+ \alpha^2(t^2+r^2)^2)(t^2+r^2+\hat{\sigma})}{\Lambda t(s^2- t^2 - r^2)}\cr
&&\times \frac{t^2+(s+ r)^2+\hat{\sigma}}{\Lambda t}\frac{t^2 + (s+ r)^2 + \alpha^2(t^2 + (s+ r)^2)^2}{t^2+ (s+ r)^2 - s^2} \cr
&<& \frac{(t^2+r^2+\hat{\sigma})(t^2+(s+ r)^2+\hat{\sigma})}{\Lambda^2t^2}\frac{(s^2/3+\alpha^2s^4/9)(5s^2/3+\alpha^2s^425/9)}{(2/3)s^2t^2}\cr
&<&\frac{25}{18}\frac{(t^2+r^2+\hat{\sigma})(t^2+(s+ r)^2+\hat{\sigma})}{\Lambda^2}\frac{(1+\alpha^2s^2)^2}{\delta^4s^2} \hbox{    for    } \alpha\geqslant0.
\end{eqnarray}
For $\alpha=0$ we can improve this estimate as
\begin{equation}\label{U'}
\frac{1}{2} < \frac{5}{6}\frac{(t^2+r^2+\hat{\sigma})(t^2+(s+ r)^2+\hat{\sigma})}{\Lambda^2}\frac{1}{\delta^4s^2}.
\end{equation}

Therefore, we obtain the left-hand side inequality in \eqref{Estimate}:
$$c_1(\alpha,t,r,s) \leqslant \hat{\sigma}(\Lambda).$$
For $r<0$ we use $d_{-1}$ instead of $d_1$.

Now for $r\geqslant 0$, we set $\hat{\sigma} = 0$ in \eqref{U}  and use the inequalities $t^2+r^2< s^2/3, \, t^2+(s+r)^2<(5/3)s^2$ to obtain the upper bound of $\lambda$ as
\begin{eqnarray*}
\Lambda &\leqslant& \frac{55}{21\sqrt 2}\left( (s^2/3)(5s^2/3)\right)^{1/2}\frac{(1+\alpha^2s^2)}{\delta^2s}\cr
&=& \frac{55\sqrt 5}{63\sqrt 2}\frac{s(1+\alpha^2s^2)}{\delta^2} \hbox{    for   } \alpha \geqslant 0.
\end{eqnarray*}
For $\alpha=0$, using \eqref{U'} we obtain that
\begin{equation*}
\Lambda \leqslant \frac{\sqrt{10}}{\sqrt 6}\left( (s^2/3)(5s^2/3)\right)^{1/2}\frac{1}{\delta^2s}=\frac{5\sqrt{2}}{3\sqrt 6}\frac{s}{\delta^2}.
\end{equation*}
\end{proof}
Rewriting \eqref{LU} in the term of $\lambda(s)$ we see that for
$$\lambda_{\alpha\geq 0} = \frac{110\sqrt{5}\pi}{63}s\delta^{-2}(1+\alpha^2s^2)^2,$$
$$\lambda_{\alpha=0} = \frac{20\pi}{3\sqrt{6}}s\delta^{-2},$$
each point in $(t,\, r)$-plane satisfying \eqref{Con} produces an unstable (positive) eigenvalue
$\hat{\sigma}>0$ of multiplicity two (the equation for the coefficients $b_k$ is the same).
Denoting by $d(s)$ the number of points of the integer lattice inside the region $A(\delta)$ we obviously have
\begin{equation}
d(s):= \sharp \left\{ (t,r)\in D(s) = \mathbb{Z}^2\cap A(\delta) \right\} \simeq a(\delta)s^2 \hbox{   as   } s \to \infty,
\end{equation}
where $a(\delta)s^2=|A(\delta)|$ is the area of the region $A(\delta)$. Therefore
the dimension of the unstable manifold around the stationary solution $\psi_s$ is
at least $2a(\delta)s^2$ and we obtain that
\begin{equation}\label{ED}
\dim \mathcal{A} \geqslant 2d(s)\simeq 2a(\delta) s^2.
\end{equation}

It is reasonable to consider two case:\\
{\bf The case $\alpha=0$.}

We have 
$$G = \lambda_{\alpha=0} s^2 =  \frac{20\pi}{3\sqrt{6}}s^3\delta^{-2}$$
and writing the estimate \eqref{ED} in terms of the Grashof number $G$ we obtain
\begin{eqnarray*}
\mathrm{dim}\mathcal{A} &\geqslant& 2a(\delta) s^2 \simeq 2\left( \frac{3\sqrt{6}}{20\pi} \right)^{2/3} a(\delta)\delta^{4/3}G^{2/3}\cr
\mathrm{dim}\mathcal{A} &\geqslant& 2\left( \frac{3\sqrt{6}}{20\pi} \right)^{2/3} (\max_{0<\delta<1/\sqrt{3}} a(\delta)\delta^{4/3}) G^{2/3} = 0,006 G^{2/3},
\end{eqnarray*}
where $\max_{0<\delta<1/\sqrt{3}} a(\delta)\delta^{4/3} = 0,012$. This is exact the lower bound obtained for the Navier-Stokes equation (see \cite{Liu,Il2004'}).\\
{\bf The case $\alpha \ll 1$.}

Here we can obtain the following lower bound for $G \thicksim (1/\alpha)^3$. Let $0<s<1/\alpha$. Then $1+\alpha^2s^2<2$ and
$$G \leqslant \frac{440\sqrt{5}\pi}{63}s^3\delta^{-2}$$
and by the same way as above we obtain that
$$\mathrm{dim}\mathcal{A} \geqslant 2\left( \frac{63}{440\sqrt{5}\pi} \right)^{2/3}(\max_{0<\delta<1/\sqrt{3}} a(\delta)\delta^{4/3}) G^{2/3} = 0,0056 G^{2/3}.$$
In particular, setting $s\simeq 1/\alpha$ we can obtain in term of $\gamma$ that
$$C_1\frac{1}{\alpha^2} \leq \mathrm{dim}\mathcal{A} \leqslant C_2\frac{1}{\alpha^2}\left( \log \frac{1}{\alpha} \right)^{1/3}.$$

\section{Inertial manifold}
The existence of an inertial manifold for the Navier-Stokes equations remains an open problem sofar. The principal reason is the nonlinear part of these equations that is very heavy to control. However, for the Bardina equations (or the orther turbulence equations such as modified-Leray-$\alpha$), one can overcome this difficulty due to the appearance of $\alpha$ which leads to control the nonlinear part of the equations. Actually, in the case of the simplified Bardina and modified-Leray-$\alpha$ equations with the periodic boundary conditions in two-dimension with periodic boundary conditions, Titi et al. \cite{Titi2014} proved the existence of inertial manifolds. Recently, the question is answered for modified-Leray-$\alpha$ equations in three-dimension by Kostiano \cite{Ko} and by Li and Sun \cite{Li}.

Beside, there are only two results about the existence of the inertial manifold on the curve spaces such as circle and two-dimensional sphere establised by Vukadinovic \cite{Vu2009a,Vu2009b} for the Smoluchowski equation. In this part, we prove the existence of an inertial manifold for the simplified Bardina model on the two-dimensional sphere $S^2$. 

We recall the definition of the inertial manifold. Consider an evolution equation on a Hilbert space $H$
endowed with the inner product $(.,.)$, and the norm $|.|$ of the form
\begin{equation}\label{abstractE}
u_t + Au = F(u).
\end{equation}
where $A$ is a positive self-adjoint linear operator with compact inverse, and $N: H \to H$ is a locally Lipschitz function.
Since $A^{-1}$ is compact, there exists a complete set of eigenfunctions $\omega_k$ for $A$,
$$A\omega_k = \lambda_k \omega_k, \, k=1,2,...$$
We arrange the eigenvalues of $A$ in a nondecreasing sequence $\lambda_1\leqslant \lambda_2\leqslant ...$ It is a well-known fact that $\lambda_k\to\infty$ as $k\to \infty$.
\begin{definition}(Inertial Manifold)
Assume that the abstract equation \eqref{abstractE} has a solution operator $S(t)$. An inertial manifold $\mathcal{M}$ is a finite-dimensional Lipschitz manifold which is positively invariant, i.e
$$S(t)\mathcal{M} \subset \mathcal{M}, \, t\geqslant 0.$$
and exponentially attracts all orbits of the flow uniformly on any bounded set $U \subset H$ of initial data, i.e
$$\mathrm{dist}(S(t)u_0,\mathcal{M}) \leqslant C_U e^{-\mu t}, \, u_0 \in U, \, t\geqslant 0.$$
\end{definition}
There are several methods for proving the existence of inertial manifolds. The vast
majority of them require some kind of Lipschitz continuity of the non-linearity $F$ and
make use of a very restrictive spectral gap property of the linear operator $A$.
\begin{theorem}\label{IMT}
Consider the abstract equation \eqref{abstractE} we assume that the non-linearity $F$ is globally Lipschitz with Lipschitz constant $L$ and the the spectral gap condition $\lambda_{n+1} - \lambda_n > 2L$ is satisfied for some $n$. Then there exists an $n$-dimensional inertial manifold over the base spanned by first $n$ eigenvectors.
\end{theorem}

On the $2$-sphere $S^2$ we have the Hodge decomposition $C^\infty(TS^2) = \left\{ \nabla\psi: \psi \in C^\infty(S^2) \right\} \oplus \left\{\curl\psi: \psi \in C^\infty(S^2) \right\}$. Therefore, by using the Helmholtz-Leray projection the simplified Bardina equation takes the form
\begin{equation}\label{BarEquOnS2}
v_t + \nu Av + B(u,u) =f,
\end{equation}
where $A= \curl\curl_n$ and $B(u,u)=-\mathbb{P}(u\times \curl_nu)$. We notice that we do not need to add the dissipative term to the equation since $\mathcal{H}^1 = \left\{ \vec{0} \right\}$. 

The $H^1$- and $H^2$-estimates are more simpler than the ones in the generalized $2$-dimensional closed manifolds obtained in Section 3.1. Indeed, we take the scalar product in $L^2(TS^2)$ of Equation \eqref{BarEquOnS2} and $u$:
$$\frac{1}{2}\frac{d}{dt} (|u|^2 + \alpha^2 \left\| u \right\|^2) + \nu (\left\| u \right\|^2 + \alpha^2| Au |^2) \leq |\left<f,u\right>|.$$
By Cauchy-Schwarz inequality, we have
\begin{equation*}
|\left< f,u \right>| \leq |A^{-1}f||Au|,
\end{equation*}
and by Young's inequality we have
\begin{equation*}
|\left< f,u\right>| \leq \frac{|A^{-1}f|^2}{2\nu\alpha^2} + \frac{\nu}{2}\alpha^2|Au|^2.
\end{equation*}
Therefore
\begin{equation*}\label{iinee1}
\frac{d}{dt}(|u|^2 + \alpha^2\left\| u \right\|^2) + \nu(\left\|u \right\|^2 + \alpha^2| Au|^2) \leq \frac{|A^{-1}f|^2}{\nu\alpha^2}.
\end{equation*}
Using Poincar\'e's and Gronwall's inequalities we obtain the $H^1$-estimate as follows
\begin{equation}\label{aabsoo1}
|u(t)|^2 + \alpha^2\left\| u(t) \right\|^2 \leq e^{-\nu\lambda_1 t}(|u_{0}|^2 + \alpha^2\left\|u_{0} \right\|^2) + \frac{|A^{-1}f|^2}{\nu^2\alpha^2\lambda_1}(1-e^{-\nu\lambda_1 t}).
\end{equation}
Taking now the inner product on $L^2(TS^2)$ of Equation \eqref{BarEquOnS2} with $Au$ with noting that $\left<B(u,u),Au\right>=0$, we get 
\begin{equation*}
\frac{1}{2}\frac{d}{dt}(\left\| u \right\|^2 + \alpha^2|Au|^2) + \nu (|Au|^2 + \alpha^2 |A^{3/2}u|^2) \leq |\left< f,Au\right>|.
\end{equation*}
Observe that by Cauchy-Schwarz and Young inequalities
\begin{equation*}
|\left< f,Au\right>| \leq |A^{-1/2}f||A^{3/2}u| \leq \frac{|A^{-1/2}f|^2}{2\alpha^2\nu} + \frac{\alpha^2\nu}{2}|A^{3/2}u|^2.
\end{equation*}
Therefore we have
\begin{equation*}
\frac{d}{dt}\left(\left\|u \right\|^2 + \alpha^2|Au|^2\right) + \nu \left( |Au|^2 + \alpha^2|A^{3/2}u|^2 \right) \leq \frac{|A^{-1/2}f|^2}{\alpha^2\nu}.
\end{equation*}
By using Poincar\'e's and Gronwall's inequalities  we obtain the $H^2$-estimate as follows
\begin{equation}\label{aabsoo2}
\left\| u(t) \right\|^2 + \alpha^2|Au(t)|^2 \leq e^{- \nu\lambda_1 t}(\left\| u(0) \right\|^2 + \alpha^2|Au(0)|^2) + \frac{|A^{-1/2}f|^2}{\nu^2\alpha^2\lambda_1}(1-e^{- \nu\lambda_1 t}).
\end{equation}

The $H^1$-estimates \eqref{aabsoo1} leads to
\begin{equation*}
\lim_{t\rightarrow\infty}|u(t)| \leq \frac{1}{2}\rho_0:= [(1+\alpha^2\lambda_1)\nu^2\alpha^2\lambda_1)]^{-1/2}|A^{-1}f|,
\end{equation*}
\begin{equation*}
\lim_{t\rightarrow\infty}\left\| u(t)\right\| \leq \frac{1}{2}\rho_1:= (\nu^2\alpha^4\lambda_1)^{-1/2}|A^{-1}f|.
\end{equation*}
Therefore, the solution $u(t)$, after long enough time, enters a ball in $H$, centered at the origin, with radius $\rho_0$. Also, $u(t)$ enters a ball in $V$ with radius $\rho_1$.

The $H^2$-estimates \eqref{aabsoo2} leads to
\begin{equation*}
\lim_{t\rightarrow\infty}\sup\left\| u(t)\right\| \leq \frac{1}{2}\tilde{\rho}_1:=[(1+\alpha^2\lambda_1)\nu^2\alpha^2\lambda_1]^{-1/2}|A^{-1/2}f|,
\end{equation*}
\begin{equation*}
\lim_{t\rightarrow\infty}\sup |Au(t)| \leq \frac{1}{2}\rho_2:= (\nu^2\alpha^4\lambda_1)^{-1/2}|A^{-1/2}f|.
\end{equation*}
We deduce that $\left\| u(t)\right\| \leq \min\left\{ \rho_1,\tilde{\rho}_1\right\}$ for $t$ large enough. Also $u(t)$ enters in the ball with radius $\rho_2$ in $D(A)$ after long enough time. 

Since $v=u+\alpha^2Au$, we have
\begin{equation*}
\lim_{t\rightarrow\infty}\sup|v(t)|\leq \lim_{t\rightarrow\infty}\sup|u(t)|+\alpha^2|Au(t)|\leq \frac{\rho_0+\alpha^2\rho_2}{2}.
\end{equation*}
Then after large time, $v(t)$ enters a ball in $H$ of the radius $\rho = \rho_0 + \alpha^2\rho_2$. Note that $\rho_0,\rho_1,\tilde{\rho}_1, \rho_2$ and $\rho$ are equivalent to $\nu^{-1}$ asymptotically.

Denoting $F(v)= -B((I+\alpha^2A)^{-1}v, (I + \alpha^2A)^{-1}v) + f = -B(u,u)+f$, then the Bardina equation \eqref{BarEquOnS2} takes the form
\begin{equation}\label{Bardina}
\frac{d}{dt}v + \nu Av = F(u) \in V', \, v(0)=v_0.
\end{equation}
The above estimates yield $u(t)\in D(A)$ hence $v(t)\in H$ for $t>0$. Since we are considering the large-time behavior of solutions, without loss of generality we can assume $v_0 \in H$. Let $v_1, v_2\in H$ then $u_1,u_2 \in D(A)$, and we have
\begin{equation*}
|Au|= |A(I+\alpha^2A)^{-1}v| \leq \frac{1}{\alpha^2}|v| \hbox{  for  } \alpha>0.
\end{equation*}
Using H\"older's inequality and Ladyzhenskaya’s inequality: $\norm{\phi}_{L^4} \leq c\norm{\phi}_{L^2}^{1/2}\norm{\nabla \phi}_{L^2}^{1/2}$, the non-linear part $B(u,v)$ can be estimated as
$$|B(u,v)| \leq c|u|^{1/2}\norm{u}^{1/2}\norm{v}^{1/2}|Av|^{1/2}.$$
Now using this estimate of $B(u,v)$ and Poincar\'e inequality, we establish that
\begin{eqnarray*}
|F(v_1)-F(v_2)|&=& |B(u_1,u_1)-B(u_2,u_2)| \cr
&=& |B(u_1,u_1-u_2)+B(u_1-u_2,u_2)| \cr
&\leq& c|u_1|^{1/2} \left\| u_1\right\|^{1/2} \left\| u_1-u_2\right\|^{1/2}|Au_1-Au_2|^{1/2} \cr
&&+ c|u_1-u_2|^{1/2}\left\| u_1-u_2\right\|^{1/2} \left\| u_2\right\|^{1/2}|Au_2|^{1/2}\cr
&\leq& c\lambda_1^{-1}(|Au_1|+|Au_2|)|Au_1-Au_2| \cr
&\leq& c\lambda_1^{-1}\alpha^{-4}(|v_1|+|v_2|)|v_1-v_2|.
\end{eqnarray*}
This yields that the nonlinear operator $F$ is locally Lipschitz from $H$ to $H$, i.e, for $v_1,\, v_2 $ in a small ball $B_\rho$ of $H$,
$$|F(v_1)-F(v_2)| \leqslant L|v_1-v_2|,$$
where $L= 2c\rho\lambda_1^{-1}\alpha^{-4}$.

We construct a prepared equation of \eqref{Bardina} as follows: let $\theta:\mathbb{R}^+ \longrightarrow [0,1]$ with $\theta(s)=1$ for $0\leq s \leq 1$, $\theta(s)=0$ for $s\geq 2$ and $\theta'(s)>2$ for $s\geq 0$. We define $\theta_\rho(s)=\theta(s/\rho)$ for $s\geq 0$ and the prepared equation of \eqref{Bardina} is given by
\begin{equation}\label{PreEq}
\frac{dv}{dt} + \nu Av = \theta_\rho(|v|)(F(v)+f):=\mathcal{F}(v).
\end{equation}

For $t$ sufficiently large, $v(t)$ enters a ball in $H$ with radius $\rho$, this leads to the fact that Equations \eqref{Bardina} and \eqref{PreEq} have the same asymptotic behaviors in time, and the same dynamics in the neighborhood of the global attractor. Furthermore \eqref{PreEq} has also an absorbing invariant ball in $H$. Indeed take the scalar product of \eqref{PreEq} with $v$, then for $|v|\geqslant 2\rho$ we have 
\begin{equation*}
\frac{d}{dt}|v|^2 + 2\nu\lambda_1|v|^2 \leq \frac{d}{dt}|v|^2 + 2\nu \left\| v\right\|^2 = 0 \hbox{  for all  } t\geqslant 0,
\end{equation*}
since $\theta_\rho(|v|) = 0$ for $|v|\geqslant 2\rho$. It follows that, if $|v_0| > 2\rho$, the orbit of the solution to \eqref{PreEq} will converge exponentially to the ball of radius $2\rho$ in $H$, while if $|v_0| \leq 2\rho$, the solution does not leave this ball. 
\begin{theorem}
The prepared equation \eqref{PreEq} of the simplified Bardina equation has an $n$-dimensional inertial manifold $\mathcal{M}$ in $H$. Furthermore, the inertial manifold $\mathcal{M}$ has the exponential tracking
property (so called normally hyperbolic inertial manifold), i.e: for any $v_0 \in H$, there exists $\phi_0 \in \mathcal{M}$ such that
$$|S(t)v_0 - S(t)\phi_0| \leq Ce^{-\mu_nt},$$
where $\mu_n \geqslant \dfrac{\lambda_{n+1}+\lambda_n}{2}$ for some $n$ and the constant $C$ depends on $|v_0|$ and $|\phi_0|$.
\end{theorem}
\begin{proof}
 The function $\mathcal{F}(v)= \theta_\rho(|v|)(F(v) +f)$ is globally Lipschitz from $H$ to $H$ due to that $F$ is locally Lipschitz and
$$|\mathcal{F}(v_1)-\mathcal{F}(v_2)| \leqslant L|v_1-v_2|, \hbox{  where  } L= 2c\rho\lambda_1^{-1}\alpha^{-4}.$$

On the $2-$sphere $S^2$, the eigenvalues of $A=\curl\curl_n$ can be calculated explicitly as $\lambda_n= n(n+1)$. Therefore we have the distance of the two successive eigenvalues
\begin{equation*}
\lim_{n\rightarrow +\infty}(\lambda_{n+1}-\lambda_n)= \lim_{n\rightarrow +\infty}[(n+1)(n+2)-n(n+1)]= \lim_{n\rightarrow +\infty}2(n+1)=  +\infty.
\end{equation*}
Hence, there exists $n$ large enough such that $\lambda_{n+1}-\lambda_n>2L$.

The non-linearity $\mathcal{F}$ of the prepared equation \eqref{PreEq} is globally Lipschitz and the operator $A=\curl\curl_n$ on $S^2$ satisfies the spectral gap condition. By applying Theorem \ref{IMT} we obtain the existence of the inertial manifold $\mathcal{M}$ for \eqref{PreEq}.

The exponential tracking property of $\mathcal{M}$ holds by using Theorem 5.2 in \cite{Fo1989} and we can show that the number $\mu_n \geqslant \dfrac{\lambda_{n+1}+\lambda_n}{2}$ from the formula of $\mu_n$ given in Theorem 4.1 in \cite{Fo1989}.
\end{proof}

\section*{Appendix}
We prove the estimate \eqref{INE} of $\norm{\rho}_\infty^{1/2}$. Indeed, we known that for $\theta \in H^3(S^2)$ and for any integer $k\geq 0$ the following inequality holds (see Lemma 4.3 in \cite{Il1994}).
\begin{equation}\label{basicE}
2\sqrt{\pi}\norm{\nabla\theta}_\infty \leqslant |\Delta\theta|(2\log(k+1)+1)^{1/2} + (k+1)^{-1}(2\lambda_1^{-1})^{1/2}|\nabla\Delta\theta|,
\end{equation}
where $|.|$ denotes the norm in $L^2$.

Let $\xi_1,...,\xi_N \in \mathbb{R}$ such that $\sum_{i=1}^N\xi^2_i =1$. We have
$$\sum_{i=1}^N \xi_iv_i = n \times \nabla\left( \Delta^{-1}\sum_{i=1}^N\xi_i\theta_i \right).$$
Using the inequality \eqref{basicE} we get that
\begin{eqnarray*}
2\sqrt{\pi}\left|\sum_{i=1}^N\xi_iv_i(s) \right| &\leqslant& 2\sqrt{\pi}\norm{\nabla\left( \Delta^{-1}\sum_{i=1}^N\xi_i\theta_i \right)}_\infty\cr
&\leqslant& \left|\sum_{i=1}^N\xi_i\theta_i\right|(2\log(k+1)+1)^{1/2} \cr
&&+ (k+1)^{-1}(2\lambda_1^{-1})^{1/2}\left|\sum_{i=1}^N\xi_i\nabla\theta_i\right|.
\end{eqnarray*}
Since $\left\{\theta_i \right\}_{i=1}^N$ are orthonormal in $\mathbb{H}$ with the norm $\norm{.}_\alpha$, we have
$$\left|\sum_{i=1}^N\xi_i\theta_i\right|^2 \leqslant \frac{1}{1+\alpha^2\lambda_1}\norm{\sum_{i=1}^N\xi_i\theta_i}^2_{\alpha} = \frac{1}{1+\alpha^2\lambda_1}\sum_{i=1}^N\xi_i^2 = \frac{1}{1+\alpha^2\lambda_1}.$$
Using the Cauchy inequaIity for the second term we obtain
\begin{eqnarray*}
\left|\sum_{i=1}^N\xi_iv_i(s) \right|^2 &=&  \left(\sum_{i=1}^N\xi_iv^1_i(s) \right)^2 + \left(\sum_{i=1}^N\xi_iv^2_i(s) \right)^2  \cr
&\leqslant& \frac{(2\sqrt{\pi})^{-2}}{1+\alpha^2\lambda_1} \left( (2\log(k+1)+1) + (k+1)^{-1}\sqrt{2} \left( \lambda_1^{-1}\sum_{i=1}^N |\nabla\theta_i|^2 \right)^{1/2} \right)^2:=c^2,
\end{eqnarray*}
where $v_i = v_i^1 + v_i^2$ is some orthogonal decomposition of $v_i(s)$ at a point $s$. 

By substituting
$$\xi_i = \frac{v_i^1}{\left(\sum_{i=1}^N(v_i^1)^2\right)^{1/2}}$$
and then 
$$\xi_i = \frac{v_i^2}{\left(\sum_{i=1}^N(v_i^2)^2\right)^{1/2}}$$
in the above inequality, we therefore obtain that
$$\rho(s)= \sum_{i=1}^N |v_i(s)|^2 \leqslant 2c^2$$
and the inequality \eqref{INE} holds.

On the generalized two dimensional closed manifold $M$ we have that for $v \in H^2(TM)$ and for any integer $k\geq 0$ the following inequality holds (see \cite{Br}).
\begin{equation*}
\norm{v}_\infty \leqslant l \left( \norm{v}_{H^1} (\log((k+1)^2+1))^{1/2} + (k+1)^{-1} \norm{v}_{H^2} \right),
\end{equation*}
Putting $v = \curl \theta$, where $\theta\in H^3(M)$, then
\begin{equation*}
\norm{\nabla\theta}_\infty \leqslant L \left( |\Delta\theta|(2\log(k+1)+1)^{1/2} + (k+1)^{-1}|\nabla\Delta\theta| \right),
\end{equation*}
By the same way as above we can obtain the inequality \eqref{INE'} as well as \eqref{INE} as follows
\begin{equation*}
\sqrt{(1+\alpha^2\lambda_1)}\norm{\rho}_\infty^{1/2} \leqslant L \left( (2\log(k+1)+1)^{1/2} + \sqrt{\lambda_1}(k+1)^{-1}\left( \lambda_1^{-1}\sum_{i=1}^N|\nabla\theta_i|^2 \right)^{1/2} \right).
\end{equation*}

\end{document}